\documentclass[prl,aps,superscriptaddress,twocolumn]{revtex4-2}

\usepackage{subfigure}
\usepackage{xcolor}
\usepackage{graphicx}
\usepackage{dcolumn}
\usepackage{bm}
\usepackage{amssymb}
\usepackage{comment}
\usepackage{physics}

\usepackage{booktabs}
\usepackage{amsthm}

\usepackage[colorlinks=true]{hyperref}
\hypersetup{linkcolor=blue,citecolor=blue,urlcolor=blue}

\newcommand{\be}{\begin{equation}}
\newcommand{\ee}{\end{equation}}

\newcommand{\bea}{\begin{eqnarray}}
\newcommand{\eea}{\end{eqnarray}}
\newcommand{\ba}{\begin{align}}
\newcommand{\ea}{\end{align}}

\newcommand{\beq}{\begin{equation}}
\newcommand{\eeq}{\end{equation}}

\newtheorem{thm}{Theorem}

\newtheorem{proposition}{Proposition}
\newtheorem{lemma}{Lemma}

\begin{document}

\title{Floquet Majorana XYZ Codes with Tunable Logical Dynamics}

\author{Xinyu Sun}
\affiliation{Institute for Advanced Study, Tsinghua University, Beijing 100084, China}

\author{Hong Yao}
\email{yaohong@tsinghua.edu.cn}
\affiliation{Institute for Advanced Study, Tsinghua University, Beijing 100084, China}

\date{\today}

\begin{abstract}
We construct Floquet codes from the Majorana XYZ subsystem code with local realizations both in qubits and directly in microscopic Majorana modes with lattice size $L\times L$.
For a three-step cycle, odd $L$ supports one static logical qubit, whereas even $L$ supports two.
For even $L$ with $L=4n+2$, both logical qubits admit time-independent Pauli representatives.
For $L=4n$, by contrast, one Pauli of the second logical qubit remains fixed throughout the cycle, while every representative of its conjugate must evolve through the measurement cycle.
This distinction follows from a parity-dependent algebraic obstruction and disappears when the cycle is reduced to two steps, which restores a fully static logical pair. 
Thus, the same encoded logical degree of freedom can be switched between static and partially dynamical forms by the measurement schedule.
With one open direction, suitable protocols can be implemented using only local Majorana parity measurements.
To our knowledge, this is the first Floquet-code family with both a local qubit representation and a direct microscopic Majorana realization.
\end{abstract}

\maketitle

\textit{Introduction}.---
Quantum error correction (QEC) protects quantum information by encoding logical degrees of freedom nonlocally and is essential for fault-tolerant quantum computation~\cite{Shor1995Scheme,Steane1996Error,Knill1997Theory,Emanuel1998Resilient}.
Topological stabilizer codes provide paradigmatic examples, combining geometrically local measurements with topological protection~\cite{bravyi1998quantum,Dennis2002Topological,Kitaev2003Fault,Bombin2006Topological,Fowler2012Surface,Bombin2015Gauge}, while subsystem codes introduce noncommuting gauge operators whose measurements can simplify syndrome extraction~\cite{Kribs2005Unified,David2006Operator,Bacon2006Operator,Bombin2010Topological,Bravyi2013subsystem}.
Floquet quantum error-correcting codes further generalize this framework by replacing a fixed stabilizer group with a periodic sequence of incompatible measurements, leading to time-dependent instantaneous stabilizer groups and logical operators~\cite{Hastings2021dynamically,vuillot2021planar,Gidney2022benchmarking,Aasen2022Adiabatic,Haah2022boundaries,Davydova2023Floquet,Zhang2023cube,ellison2023floquet,Sullivan2023Floquet,Davydova2024quantum,Dua2024Engineering,Kesselring2024Anyon,Kobayashi2024Cross,Higgott2024Constructions,Fu2025error,watanabe2026floquet,zen2026low,capatos2026dynamic,Rodatz2026floquetifying}.

The implementation of a quantum code depends strongly on the native operations of the hardware.
Surface, subsystem, and Floquet codes have been explored in superconducting circuits~\cite{Andersen2020Repeated,Krinner2022Realizing,Google2023Suppressing}, trapped ions~\cite{Ryan2021Realization,Egan2021Fault,Postler2022Demonstration}, neutral-atom arrays~\cite{Bluvstein2023Logical,Xu2024constant,Bluvstein2026Fault}, and Majorana-based architectures~\cite{Vijay2015Majorana,Paetznick2023Performance}.
Majorana platforms are natural for measurement-based codes because local fermion-parity measurements are native operations~\cite{Roadmap2016Plugge,Karzig2017Scalable,Litinski2018Quantum,Viyuela2019Scalable,microsoft2025interferometric,Mudassar2026Fault}.
Recently, the Majorana XYZ subsystem code was introduced from a lattice of Majorana modes with an effective qubit description in terms of noncommuting $XYZ$ gauge generators~\cite{Busse2026Majorana}.
This raises the question of what logical structure emerges under periodic measurements and whether the resulting Floquet code can be realized locally both in qubits and directly in microscopic Majorana modes.


In this letter, we construct and exactly characterize one-, two-, and three-step measurement protocols for the Majorana XYZ subsystem code~\cite{Busse2026Majorana}, revealing a parity- and schedule-dependent coexistence of static and dynamical logical operators.
We consider an $L\times L$ lattice.
For the three-step protocol, odd $L$ supports one static logical qubit, while even $L$ supports two.
Both are static for $L=4n+2$, whereas for $L=4n$ one Pauli of the second logical qubit is static, but its conjugate cannot be chosen to be time independent. 
Reducing the period to two steps removes this obstruction and restores two static logical qubits. 
We further obtain a direct local realization in microscopic Majorana modes with one open spatial direction.
To our knowledge, this is the first Floquet-code family with both a local qubit representation and a direct microscopic Majorana realization.

\begin{figure}[t]
\centering
\includegraphics[width=0.8\columnwidth]{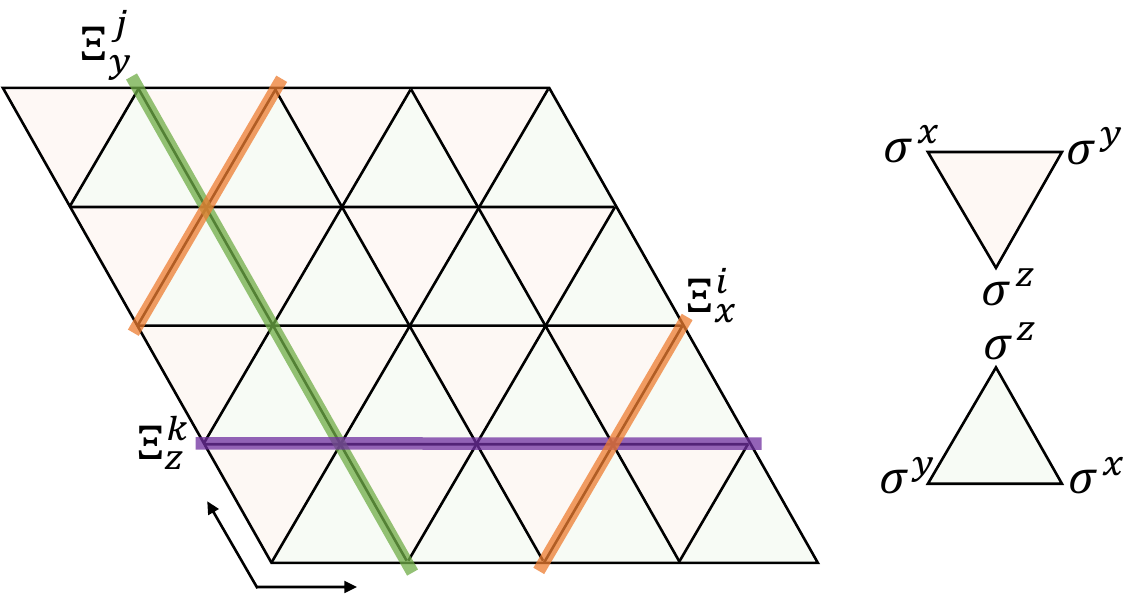}
\caption{
Static Majorana XYZ subsystem code.
The two inequivalent triangle operators on the right define the gauge generators $A_{\Delta}$ and $B_{\nabla}$.
The colored lines on the triangular lattice show representative noncontractible loop operators $\Xi_x$, $\Xi_y$, and $\Xi_z$.
Parallel loops differ by stabilizers and therefore provide equivalent logical representatives.
}
\label{fig:MajoranaXYZ}
\end{figure}

\textit{Static Majorana XYZ Code and Measurement Algebra}.---
We consider $N=L^2$ qubits on the periodic triangular lattice inherited from the Majorana construction. 
The Pauli ordering of the two inequivalent triangles and representative noncontractible loops are shown in Fig.~\ref{fig:MajoranaXYZ}. 
The gauge group is generated by
\begin{equation}
    A_{\Delta}=\sigma_i^x\sigma_j^y\sigma_k^z,\quad
    B_{\nabla}=\sigma_i^z\sigma_j^y\sigma_k^x,\quad
    \mathcal{G}=\langle A_{\Delta},B_{\nabla}\rangle .
\label{eq:gauge_group}
\end{equation}
For later use, we denote by $\Xi_{\alpha,a}^j$ a loop operator specified by three labels.
Here $\alpha=h,v,d$ denotes the line orientation (horizontal, vertical, or diagonal), $a=x,y,z$ the Pauli type, and $j$ the translated noncontractible line within that orientation.
The three loop families naturally associated with the static code are abbreviated as
\begin{equation}
\Xi_x^j\equiv\Xi_{d,x}^j,\qquad
\Xi_y^j\equiv\Xi_{v,y}^j,\qquad
\Xi_z^j\equiv\Xi_{h,z}^j ,
\label{eq:loop_notation}
\end{equation}
as illustrated in Fig.~\ref{fig:MajoranaXYZ}.
When the line index is unimportant, we write simply $\Xi_x$, $\Xi_y$, and $\Xi_z$.

\textit{Proposition 1 (Static Majorana XYZ code).}
For any $L$, the gauge group in Eq.~\eqref{eq:gauge_group} and its stabilizer subgroup $\mathcal{S}=Z(\mathcal{G})$ satisfy
\begin{equation}
    \operatorname{rank}\mathcal{G}=(2L-1)(L-1),\qquad
    \operatorname{rank}\mathcal{S}=3(L-1).
\label{eq:static_rank}
\end{equation}
Consequently, the static Majorana XYZ code has parameters $[[L^2,1,(L-1)(L-2),L]]$, where $[[n,k,r,d]]$ denotes a subsystem code with $n$ physical qubits, $k$ logical qubits, $r$ gauge qubits, and distance $d$.

The stabilizer group is generated by neighboring pairs of parallel loops, $D_x^i=\Xi_x^i\Xi_x^{i+1}$, $D_y^j=\Xi_y^j\Xi_y^{j+1}$, and $D_z^k=\Xi_z^k\Xi_z^{k+1}$.
For odd $L$, these give $3(L-1)$ independent generators.
For even $L$, they obey one additional global relation, $\prod_{i\,{\rm even}}D_x^i\prod_{j\,{\rm even}}D_y^j\prod_{k\,{\rm even}}D_z^k=I$, so only $3(L-1)-1$ double loops are independent.
The remaining stabilizer can be chosen as the three-loop operator
\begin{equation}
    H_L=\Xi_x^L\Xi_y^L\Xi_z^L .
\label{eq:HL}
\end{equation}
For odd $L$, $H_L$ is generated by the double loops, whereas for even $L$ it is independent of them.
Thus $\operatorname{rank}\mathcal{S}=3(L-1)$ in both cases, but the stabilizer generating structure depends on the parity of $L$.

A convenient choice of logical Pauli operators is
\begin{equation}
    \widetilde X=\Xi_x,\qquad
    \widetilde Z=\Xi_z ,
\label{eq:static_logical}
\end{equation}
where any intersecting pair of $x$- and $z$-type loops may be chosen.
Neighboring parallel loops differ by double-loop stabilizers, giving $L$ mutually disjoint representatives of each of these logical classes.
Thus any conjugate logical operator must intersect all of them and therefore has weight at least $L$.
Since a noncontractible loop itself has weight $L$, the distance is $d=L$.
The detailed rank counting, stabilizer relations, logical operators, and distance proof are given in the Supplemental Material (SM)~\cite{SM}.

We next introduce the commuting gauge operators used in the
measurement protocols below.
The product of a neighboring pair of triangle generators gives a four-qubit parallelogram of a single Pauli type,
$P_{a,p}=\prod_{r\in p}\sigma_r^a$ with $a=x,y,z$, where $p$ labels the corresponding parallelogram shown in Fig.~2(a).
For each Pauli type, all translated parallelograms mutually commute.
We denote the groups they generate by $\mathcal P_x,\mathcal P_y$, and $\mathcal P_z$.
Let $\mathcal D_a=\langle D_a^j\rangle$ denote the subgroup generated by the $a$-type double loops.

\textit{Proposition 2 (Measurement algebra).---}
For each $a=x,y,z$, $\mathcal P_a$ is an Abelian subgroup of
$\mathcal G$ with $\operatorname{rank}\mathcal P_a=(L-1)^2$,
and contains the $L-1$ independent same-type double loops $\mathcal D_a$.
For any cyclic permutation $(a,b,c)$ of $(x,y,z)$,
\begin{equation}
    \left\langle\mathcal P_b,\, \mathcal P_a\cap C(\mathcal P_b)\right\rangle=
    \begin{cases}
        \langle\mathcal P_b,\mathcal D_a\rangle,
            & L\ {\rm odd},\\[2pt]
        \langle\mathcal P_b,\mathcal D_a,\mathcal D_c\rangle,
            & L\ {\rm even}.
    \end{cases}
\label{eq:measurement_centralizer}
\end{equation}
The explicit generating sets, parity-dependent centralizer
calculation, and proof are given in the SM~\cite{SM}.

\begin{figure}[t]
    \centering
    \includegraphics[width=\columnwidth]{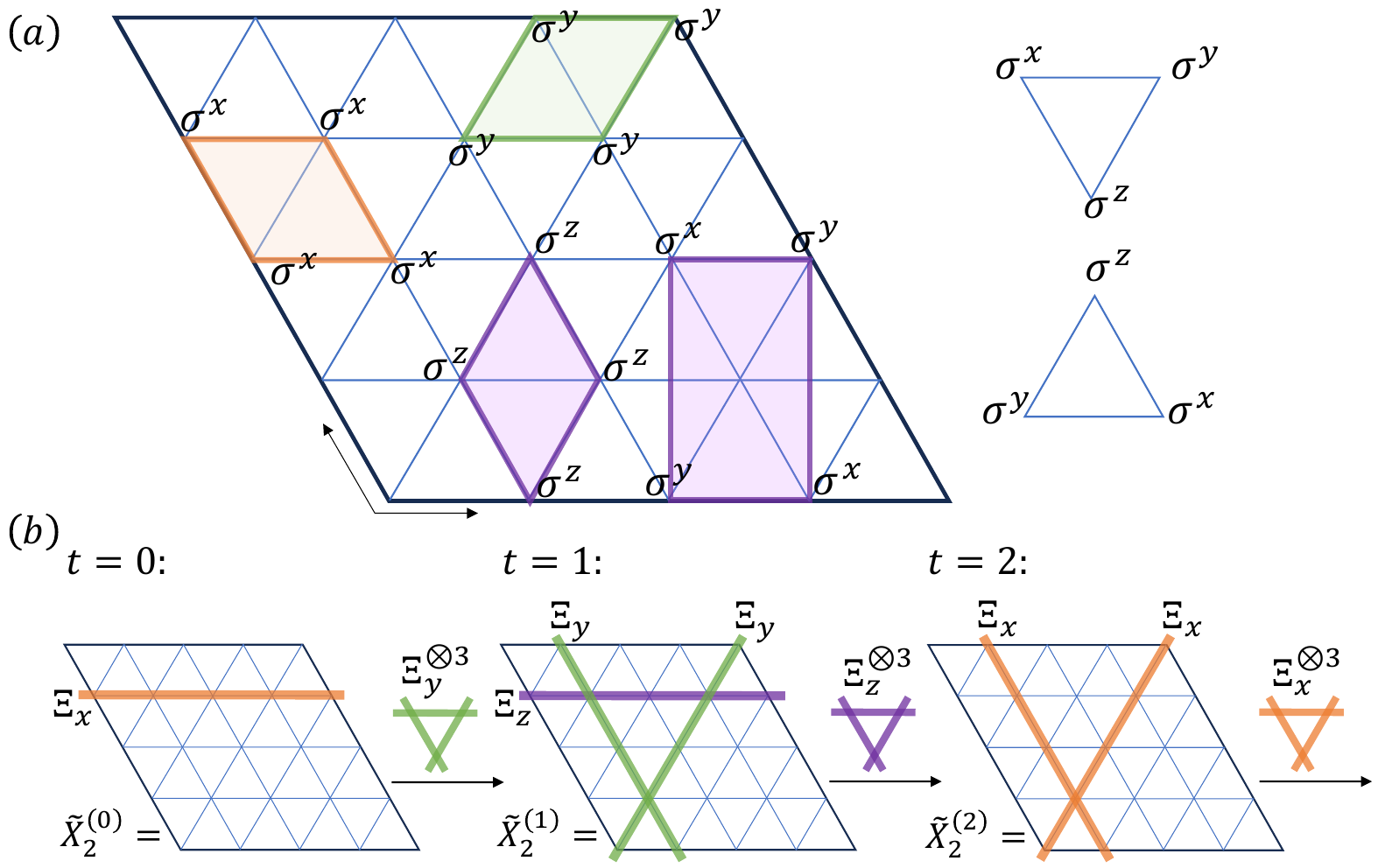}
    \caption{
    Three-step measurement protocol and dynamical logical operator.
    (a) The three families of four-qubit parallelogram operators, $\mathcal{P}_x$, $\mathcal{P}_y$, and $\mathcal{P}_z$, measured successively in the $T=3$ protocol.
    The Pauli labels indicate the corresponding single-Pauli-type operators on each parallelogram.
    (b) Evolution of the second logical Pauli representative for $L=4n$.
    Here $\Xi_a^{\otimes 3}$ denotes the shorthand $\Xi_{h,a}\Xi_{v,a}\Xi_{d,a}$.
    The logical operator $\widetilde Z_2=H_L$ remains time independent throughout the cycle, whereas its conjugate representative evolves as $\widetilde X_2^{(0)}\rightarrow\widetilde X_2^{(1)}\rightarrow\widetilde X_2^{(2)}$ under successive measurements of $\mathcal{P}_x$, $\mathcal{P}_y$, and $\mathcal{P}_z$.
    The three representatives describe the same encoded logical Pauli operator at the corresponding instantaneous stabilizer groups.
    }
    \label{fig:measurement}
\end{figure}

\textit{Three-Step Floquet Majorana XYZ Code}.---
We now promote the static subsystem code to a Floquet code by successively measuring the three parallelogram families shown in Fig.~\ref{fig:measurement}(a),
\begin{equation}
    t=0:\ \mathcal{P}_x,\qquad
    t=1:\ \mathcal{P}_y,\qquad
    t=2:\ \mathcal{P}_z,
\label{eq:T3_protocol}
\end{equation}
with period $T=3$.
Ignoring measurement-outcome signs, the instantaneous stabilizer group (ISG) updates as $\mathcal{S}_t=\langle \mathcal{P}_{a_t},\,\mathcal{S}_{t-1}\cap C(\mathcal{P}_{a_t})\rangle$~\cite{Hastings2021dynamically,Davydova2023Floquet}, where $C(\mathcal{P}_{a_t})$ denotes the Pauli centralizer of the newly measured family.

\textit{Theorem 1 (Three-step protocol).}
After the initial transient, the protocol in Eq.~\eqref{eq:T3_protocol} has distance $d=L$ and the following logical structure:
for odd $L$, it encodes one static logical qubit;
for $L=4n+2$, it encodes two static logical qubits;
and for $L=4n$, it encodes one static and one partially dynamical logical qubit.

The approach to the steady Floquet cycle depends crucially on the parity of $L$.
Let $\mathcal{D}_a=\langle D_a^j\rangle$ denote the subgroup generated by the $a$-type double loops.
Starting from a trivial ISG, the first measurement gives $\mathcal{S}_0=\mathcal{P}_x$ for either parity.

For odd $L$, the first cycle evolves as
\begin{equation}
    \mathcal{S}_0=\mathcal{P}_x,\quad
    \mathcal{S}_1=\langle\mathcal{P}_y,\mathcal{D}_x\rangle,\quad \mathcal{S}_2=\langle\mathcal{P}_z,\mathcal{D}_x,\mathcal{D}_y\rangle .
\label{eq:T3_ISG_odd}
\end{equation}
Accordingly, $\operatorname{rank}\mathcal{S}_0=(L-1)^2$ and $\operatorname{rank}\mathcal{S}_1=L^2-L$, while the third measurement brings the ISG to $\operatorname{rank}\mathcal{S}_2=L^2-1$.
From the next round onward, the three ISGs are related by cyclic permutation of $x,y,z$ and retain rank $L^2-1$, leaving a single logical qubit.
It is inherited from the static subsystem code and admits the time-independent logical pair $\widetilde X_1=\Xi_x$ and $\widetilde Z_1=\Xi_z$, which commutes with all three measurement families.
Thus, for odd $L$, the third measurement family is essential for reaching the one-logical-qubit Floquet code.

For even $L$, the second measurement behaves qualitatively differently.
The intersection $\mathcal{P}_x\cap C(\mathcal{P}_y)$ is enlarged by the even-$L$ global relations, and an equivalent generating set for the first-cycle ISGs is
\begin{equation}
    \mathcal{S}_0=\mathcal{P}_x,\quad \mathcal{S}_1=\langle\mathcal{P}_y,\mathcal{D}_x,\mathcal{D}_z\rangle,\quad \mathcal{S}_2=\langle\mathcal{P}_z,\mathcal{D}_x,\mathcal{D}_y\rangle .
\label{eq:T3_ISG_even}
\end{equation}
Because the three double-loop families obey one additional global relation, already $\operatorname{rank}\mathcal{S}_1=L^2-2$; the subsequent $\mathcal{P}_z$ measurement changes the generating set but not the rank.
Hence, the two-logical-qubit structure is reached after the first two measurements and persists throughout the steady cycle.
The first logical qubit may again be chosen as $\widetilde X_1=\Xi_x$ and $\widetilde Z_1=\Xi_z$.
The second arises from the static stabilizer $H_L$, which is absent from the Floquet ISG but commutes with it, and we choose $\widetilde Z_2=H_L$.
The explicit centralizer calculations and independent ISG generating sets are given in the SM~\cite{SM}. 

The nature of the conjugate $\widetilde X_2$ depends on $L$ modulo four.
For $L=4n+2$, a time-independent choice is
\begin{equation}
\widetilde X_2=\Omega_x\equiv\prod_{i\,{\rm even}}\Xi_{h,x}^i,\qquad
\widetilde Z_2=H_L .
\label{eq:T3_static_second_logical}
\end{equation}
The operator $\Omega_x$ commutes with all three measurement families and satisfies
$\{\Omega_x,H_L\}=0$.
Hence both qubits admit time-independent logical Pauli pairs for $L=4n+2$.

For $L=4n$, by contrast, no time-independent conjugate of $H_L$ exists.
Let $\mathcal{M}_{\rm cyc}=\langle\mathcal{P}_x,\mathcal{P}_y,\mathcal{P}_z\rangle$ denote the group generated by all measurements in one period.
Since $H_L\in\mathcal{M}_{\rm cyc}$, any operator commuting with every measurement must also commute with $H_L$ and therefore cannot serve as its conjugate logical Pauli.
Nevertheless, a valid conjugate exists at every measurement round.
Defining $F_a=\Xi_{h,a}\Xi_{v,a}\Xi_{d,a}\in\mathcal{P}_a$, the physical representative evolves as
\begin{equation}
    \widetilde X_2^{(0)}\xrightarrow{\;F_y\;}\widetilde X_2^{(1)}\xrightarrow{\;F_z\;}\widetilde X_2^{(2)}\xrightarrow{\;F_x\;}\widetilde X_2^{(0)} .
\label{eq:X2_orbit}
\end{equation}
A convenient choice is $\widetilde X_2^{(0)}=\Xi_{h,x}$, $\widetilde X_2^{(1)}=\Xi_{h,z}\Xi_{v,y}\Xi_{d,y}$, and $\widetilde X_2^{(2)}=\Xi_{v,x}\Xi_{d,x}$, up to irrelevant Pauli phases.
Each representative commutes with the corresponding ISG and anticommutes with the fixed logical operator $\widetilde Z_2=H_L$.
Thus successive measurements necessarily update the representative of the same logical Pauli, as illustrated in Fig.~\ref{fig:measurement}(b).
We refer to this asymmetric structure, with one logical Pauli remaining static while its conjugate necessarily evolves, as a partially dynamical logical qubit.
Related static and dynamical logical operators occur in the honeycomb and ladder Floquet codes~\cite{Hastings2021dynamically}. 
Here, the obstruction to a static conjugate depends on $L\bmod 4$, occurring for $L=4n$ but not for $L=4n+2$.
The proof that $H_L\in\mathcal{M}_{\rm cyc}$ and the detailed measurement identities underlying Eq.~\eqref{eq:X2_orbit} are given in the SM~\cite{SM}.

Although the operators $H_L$ and $\Omega_x$ make the static and dynamical character of the logical qubits transparent, simpler representatives can be chosen independently at each measurement round.
For even $L$, a convenient round-adapted logical basis is
\begin{equation}
\begin{split}
(\widehat X_1,\widehat Z_1)
&:\ (\Xi_{v,y},\Xi_{h,z})
\rightarrow(\Xi_{h,z},\Xi_{d,x})
\rightarrow(\Xi_{d,x},\Xi_{v,y}),\\
(\widehat X_2,\widehat Z_2)
&:\ (\Xi_{v,z},\Xi_{h,y})
\rightarrow(\Xi_{h,x},\Xi_{d,z})
\rightarrow(\Xi_{d,y},\Xi_{v,x}),
\end{split}
\label{eq:T3_distance_logicals}
\end{equation}
where the three entries correspond to the $\mathcal{P}_x$, $\mathcal{P}_y$, and $\mathcal{P}_z$ rounds, respectively.
At each round, these operators form a complete basis for the two encoded qubits and all have weight $L$~\footnote{The round-adapted basis is chosen independently at each measurement round and is used primarily for the distance argument. Its individual entries therefore need not represent the same logical Pauli class across different rounds, and the corresponding logical frame can mix the two encoded qubits. The relation to the continuously tracked logical operators is given explicitly in the SM~\cite{SM}.}.
Following the same disjoint-representative argument as for the static code gives $d=L$ throughout the steady $T=3$ cycle~\footnote{Here and below, $d$ denotes the instantaneous code distance at each steady round. The spacetime fault distance of the full measurement protocol~\cite{Fu2025error,blackwell2025code} is not addressed here.}.
The detailed equivalence and distance proof are given in the SM~\cite{SM}.

\textit{Two- and One-Step Measurement Protocols}.---
The parity-dependent ISG evolution above suggests a natural shortening of the measurement cycle.
For odd $L$, after the first two measurements the ISG has rank only $L^2-L$, and the third family $\mathcal{P}_z$ is required to reach the rank-$(L^2-1)$ one-logical-qubit code.
For even $L$, by contrast, the rank-$(L^2-2)$ two-logical-qubit structure is already reached after measuring $\mathcal{P}_x$ and $\mathcal{P}_y$.
This motivates closing the cycle after the second measurement.

\textit{Proposition 3 (Two-step protocol).}
For even $L$, consider the measurement cycle
\begin{equation}
t=0:\ \mathcal{P}_x,\qquad
t=1:\ \mathcal{P}_y,
\label{eq:T2_protocol}
\end{equation}
repeated with period $T=2$.
After the initial transient, it encodes two static logical qubits with distance $d=L$.

Starting from $\mathcal S_0=\mathcal P_x$, the next $\mathcal P_y$ measurement produces $\mathcal S_1=\langle\mathcal P_y,\mathcal D_x,\mathcal D_z\rangle$ of rank $L^2-2$, and the following $\mathcal P_x$ round gives the cyclic partner $\langle\mathcal P_x,\mathcal D_y,\mathcal D_z\rangle$.
These two ISGs then alternate periodically.
The two steady ISGs preserve the logical qubit inherited from the static subsystem code, so the first logical pair may again be chosen as
$\widetilde X_1=\Xi_x$ and $\widetilde Z_1=\Xi_z$.
For the second logical qubit, we choose $\widetilde Z_2=H_L$ and
\begin{equation}
    \widetilde X_2=
    \begin{cases}
    \Omega_x, & L=4n+2,\\
    \Xi_{h,x}, & L=4n,
    \end{cases}
\label{eq:T2_logicals}
\end{equation}
where $\Omega_x=\prod_{i\,{\rm even}}\Xi_{h,x}^i$.
In the corresponding size class, $\widetilde X_2$ commutes with both measurement families and anticommutes with $H_L$, so both logical qubits admit time-independent Pauli representatives.

The $L=4n$ case makes the role of the measurement schedule particularly transparent.
In the three-step protocol, inclusion of $\mathcal{P}_z$ gives $H_L\in\langle\mathcal{P}_x,\mathcal{P}_y,\mathcal{P}_z\rangle$ and obstructs any time-independent conjugate of $H_L$.
For the two-step cycle, $H_L\notin\langle\mathcal{P}_x,\mathcal{P}_y\rangle$, and the static choice $\widetilde X_2=\Xi_{h,x}$ becomes possible.
Thus, changing the measurement schedule converts the same second logical qubit from partially dynamical at $T=3$ to fully static at $T=2$. 
Together with the $L\bmod 4$ dependence, this shows that its static or dynamical character is controlled by the measurement schedule and global lattice relations, rather than fixed by the underlying subsystem code.
Following the same disjoint-representative argument as above gives $d=L$.
The explicit ISG derivation and distance proof are given in the SM~\cite{SM}.

We finally consider a one-step protocol.
Rather than repeatedly measuring one parallelogram family, we use a different commuting family obtained from neighboring products of triangle gauge generators.
The resulting four-qubit rectangle operator, shown in Fig.~\ref{fig:measurement}, has the form $R_p=\sigma_i^x\sigma_j^y\sigma_k^x\sigma_l^y$.
All translated rectangle operators mutually commute, so the full family can be measured in a single round and repeated with period $T=1$.

\begin{table}[t]
\caption{
Logical structure of the one-, two-, and three-step protocols.
Each entry gives $(k_{\rm s},k_{\rm p};d)$.
A dash denotes the odd-$L$ two-step case not considered as part of this code family.
}
\label{tab:protocol_summary}
\begin{ruledtabular}
\begin{tabular}{c c c c}
System size & $T=1$ & $T=2$ & $T=3$ \\
\hline
$L$ odd
& $(1,0;L)$
& ---
& $(1,0;L)$ \\
$L=4n+2$
& $(4,0;L/2)$
& $(2,0;L)$
& $(2,0;L)$ \\
$L=4n$
& $(4,0;L/2)$
& $(2,0;L)$
& $(1,1;L)$
\end{tabular}
\end{ruledtabular}
\end{table}

\textit{Proposition 4 (One-step protocol).}
For odd $L$, the $T=1$ rectangle code encodes one static logical qubit with distance $d=L$.
For even $L$, it encodes four static logical qubits with distance $d=L/2$.

For odd $L$, the $L^2$ rectangle stabilizers obey one global relation, giving $\operatorname{rank}\mathcal{S}=L^2-1$ and hence one logical qubit.
A convenient logical pair is $\widetilde X=\Xi_x$ and $\widetilde Z=\Xi_y$.
The resulting stabilizer code is a twisted Wen-plaquette realization~\cite{Wen2003Quantum,Yu2007Exact,Yu2013Majorana}, and the same disjoint-representative argument gives $d=L$.

For even $L$, the rectangle operators separate into two independent sublattices, each realizing a Wen-plaquette sector with two independent global relations.
Thus the full stabilizer group has rank $L^2-4$ and encodes four logical qubits.
Within each sublattice, the noncontractible loop operators split into their even- and odd-line restrictions, which we denote by $\Xi_x^{\rm even},\Xi_x^{\rm odd},\Xi_y^{\rm even},\Xi_y^{\rm odd}$, giving two logical qubits of that sector.
Each restricted loop has weight $L/2$, and the same disjoint-representative argument gives $d=L/2$.
The explicit global relations and logical Pauli pairs are given in the SM~\cite{SM}.

The logical structures of the three measurement protocols are summarized in Table~\ref{tab:protocol_summary}.
We characterize each code by $(k_{\rm s},k_{\rm p};d)$, where $k_{\rm s}$ and $k_{\rm p}$ denote the numbers of static and partially dynamical logical qubits, respectively, and $d$ is the code distance.

\begin{figure}[t]
    \centering
    \includegraphics[width=0.88\columnwidth]{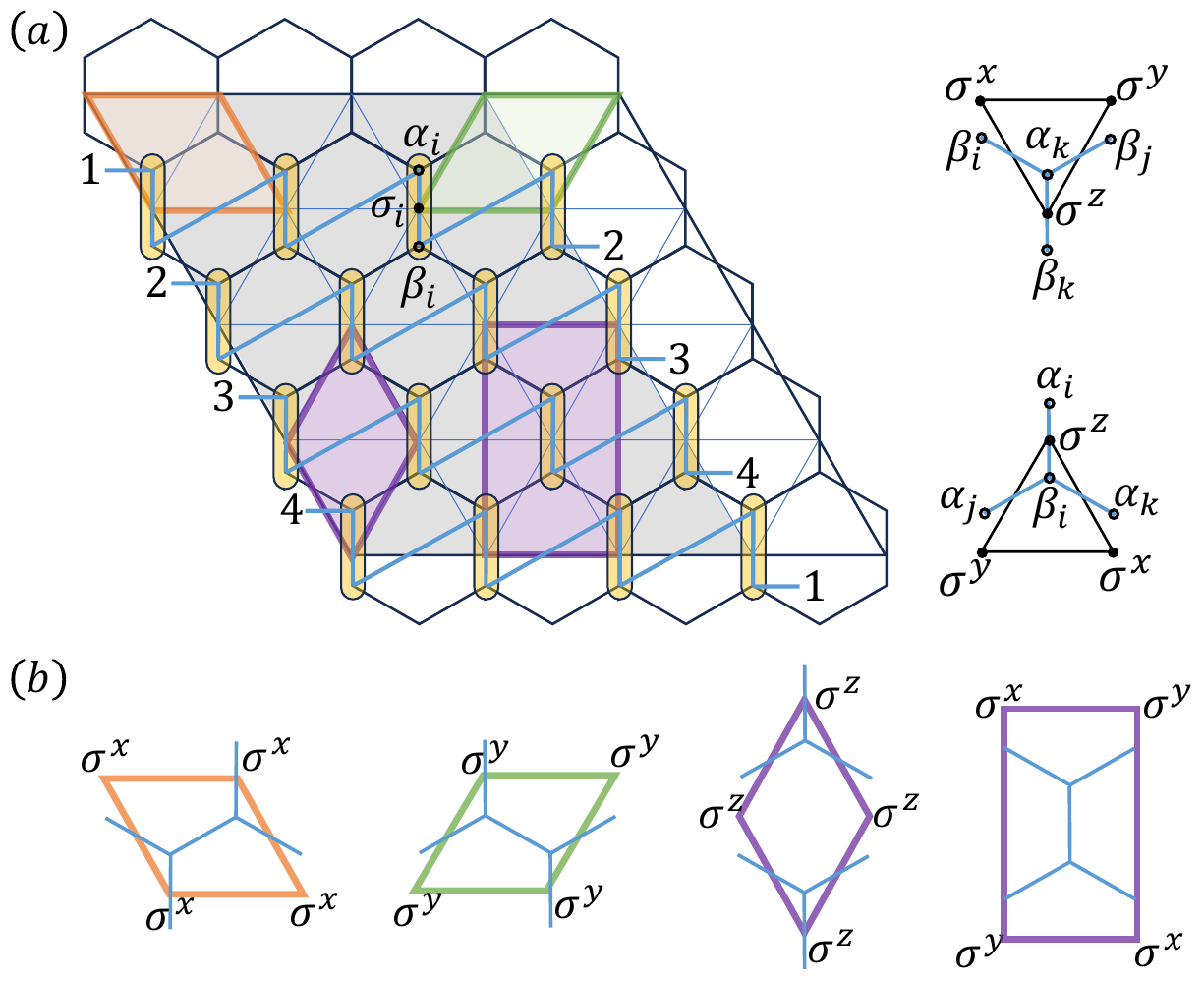}
    \caption{
    Local Majorana realization with one open spatial direction.
    (a) Honeycomb-lattice Majorana modes ordered along the indicated path.
    Each yellow region denotes a unit cell containing $\alpha_i$ and $\beta_i$, related to the effective spin $\sigma_i$ through Eq.~\eqref{eq:Majorana_spin_mapping}.
    The two retained triangle generators admit local four-Majorana representations, shown on the right.
    (b) Spin representations of the $x$-, $y$-, and $z$-type parallelogram measurements and the rectangle measurement.
    Under the Majorana mapping, the $x$- and $y$-type parallelograms and the rectangle are local four-Majorana parity operators, whereas the $z$-type one is an eight-Majorana parity operator.}
    \label{fig:Majorana_rep}
\end{figure}

\textit{Open-Boundary Majorana Realization}.---
So far we have formulated the code in terms of effective spins on the triangular lattice.
We now show that the same construction admits a direct local realization in Majorana modes when one spatial direction is opened.
As illustrated in Fig.~\ref{fig:Majorana_rep}, each unit cell of the honeycomb lattice contains two Majorana modes, $\alpha_i$ and $\beta_i$, related to the effective spin through the ordering-dependent transformation~\cite{Li2018Majorana}
\begin{equation}
    \alpha_i=\sigma_i^x\prod_{k<i}(-\sigma_k^z),\qquad
    \beta_i=\sigma_i^y\prod_{k<i}(-\sigma_k^z).
\label{eq:Majorana_spin_mapping}
\end{equation}
For the two triangle generators in Eq.~\eqref{eq:gauge_group}, the Jordan--Wigner strings cancel locally, giving
\begin{equation}
    A_{\Delta}=\beta_i\alpha_i\alpha_j\alpha_k, \qquad B_{\nabla}=\alpha_k\beta_i\beta_j\beta_k ,
\label{eq:Majorana_triangle}
\end{equation}
with the vertex convention shown in Fig.~\ref{fig:Majorana_rep}(a).
Thus each retained triangle gauge generator has a local four-Majorana representation.

On a torus, generators crossing the ordering seam generally acquire nonlocal strings.
We therefore open one direction and use the ordering shown in Fig.~\ref{fig:Majorana_rep}.
The cylindrical geometry retains $L(L-1)$ generators of each triangle type, all with local four-Majorana representations.
The omitted boundary-crossing generators are generated by products of the retained ones, so the local generating set reproduces the same gauge group, and hence the same subsystem-code structure, as the periodic spin construction.
Opening the boundary changes the microscopic realization without changing the code algebra.



Hence, suitable Floquet protocols admit both local qubit and direct microscopic Majorana realizations.
Unlike previous Majorana implementations based on tetron-encoded qubits~\cite{Hastings2021dynamically,Paetznick2023Performance}, here the Majorana modes themselves constitute the microscopic degrees of freedom.
As shown in Fig.~\ref{fig:Majorana_rep}(b), local $\mathcal P_x$ and $\mathcal P_y$ parallelograms are four-Majorana parity measurements, whereas $\mathcal P_z$ requires eight Majoranas.
For even $L$, the $T=3$ protocol remains locally realizable because the retained local $\mathcal P_y$ measurements, together with the $y$-type double loops from preceding rounds, reproduce the full $\mathcal P_y$-round ISG. 
The $T=2$ protocol can similarly be realized using local $\mathcal P_x$ and $\mathcal P_z$ measurements~\footnote{For odd $L$, the truncated local $\mathcal P_y$ family does not reproduce the periodic $T=3$ ISG because the missing $y$-type double-loop sector is not inherited from the preceding $\mathcal P_z$ and $\mathcal P_x$ rounds. Likewise, simply truncating the rectangle family does not reproduce the periodic $T=1$ code.}.
Details of the Majorana mapping and boundary ISGs are given in the SM~\cite{SM}.

\textit{Summary and Outlook}.---
We have constructed one-, two-, and three-step Floquet codes from the Majorana XYZ subsystem code, revealing parity- and schedule-dependent logical dynamics.
For $L=4n+2$, both logical qubits of the three-step protocol are static, whereas for $L=4n$ one Pauli of the second logical qubit remains static while its conjugate necessarily evolves.
Reducing the protocol to $T=2$ removes this obstruction and restores a static logical pair.
We further establish direct local Majorana realizations of the even-$L$ two- and three-step protocols with one open spatial direction, using local four- and eight-Majorana parity measurements.
To our knowledge, this is the first Floquet-code family with both a local qubit representation and a direct microscopic Majorana realization.

Several questions now become natural.
Most importantly, the fault-tolerant properties of these protocols remain to be established under realistic data and measurement errors, including decoding thresholds for the time-dependent ISGs and logical representatives~\cite{Gidney2021fault,Setiawan2025Tailoring,Fahimniya2025fault,Tang2025Phases,Derks2026dynamical}.
A second direction is to exploit the measurement schedule itself as a logical-control resource: switching between the one-, two-, and three-step cycles may provide a way to create, remove, or manipulate dynamical logical degrees of freedom~\cite{Aasen2022Adiabatic,Davydova2024quantum,Kobayashi2024Cross,Sun2025Logical,Alam2025Dynamical}. 
It would also be interesting to construct suitable boundary checks for the open geometry, potentially broadening the direct Majorana realization and facilitating implementations in experimentally motivated architectures~\cite{Ivanov2001Non,Fu2008Superconducting,Wang2018Evidence,Machida2019Zero,Liu2019Protocol,Li2022Ordered,vuillot2021planar,Haah2022boundaries}.

\textit{Acknowledgments}.---This work is supported in part by the NSFC under Grant Nos.~12347107 and 12334003 (X.S. and H.Y.), and the New Cornerstone Science Foundation through the Xplorer Prize (H.Y.).

\textit{Note added}. --- Close to this manuscript being finalized, a related work~\cite{grover2026dynamical} appeared, which independently studied the chiral XYZ triangle-operator algebra and its subsystem symmetries. 


\let\oldaddcontentsline\addcontentsline
\renewcommand{\addcontentsline}[3]{}
\bibliography{refs.bib}
\let\addcontentsline\oldaddcontentsline

\clearpage

\setcounter{section}{0}
\setcounter{subsection}{0}
\setcounter{secnumdepth}{2}
\setcounter{tocdepth}{2}

\renewcommand{\thesection}{\Roman{section}}
\renewcommand{\thesubsection}{\Alph{subsection}}
\renewcommand{\theHsection}{SM.\arabic{section}}
\renewcommand{\theHsubsection}{SM.\arabic{section}.\arabic{subsection}}

\setcounter{table}{0}
\renewcommand{\thetable}{S\arabic{table}}
\renewcommand{\theHtable}{S.\arabic{table}}

\setcounter{figure}{0}
\renewcommand{\thefigure}{S\arabic{figure}}
\renewcommand{\theHfigure}{S.\arabic{figure}}

\setcounter{equation}{0}
\renewcommand{\theequation}{S\arabic{equation}}
\renewcommand{\theHequation}{S.\arabic{equation}}

\begin{widetext}

\begin{center}
\textbf{Supplemental Material for ``Floquet Majorana XYZ Codes with Tunable Logical Dynamics''}
\end{center}

This Supplemental Material provides the detailed derivations supporting the main text.
We establish the static gauge and stabilizer structure, the parallelogram measurement algebra, and the parity-dependent instantaneous stabilizer groups for the three-, two-, and one-step protocols.
We also give the proof of the partially dynamical logical qubit for $L=4n$, the corresponding logical-operator evolution and distance, and the open-boundary Majorana realization.

\tableofcontents

\section{Static Majorana XYZ subsystem code}
\label{sec:SMstatic}

As defined in the main text, the gauge group is generated by the two oriented triangle operators
\begin{equation}
A_{\Delta}=\sigma_i^x\sigma_j^y\sigma_k^z,\qquad
B_{\nabla}=\sigma_i^z\sigma_j^y\sigma_k^x,\qquad
\mathcal G=\langle A_{\Delta},B_{\nabla}\rangle ,
\label{eq:SMtriangles}
\end{equation}
with the vertex convention of Fig.~1 in the main text.
Its stabilizer group is the center $\mathcal S=Z(\mathcal G)$.

\begin{figure}[t]
    \centering
    \includegraphics[width=0.4\columnwidth]{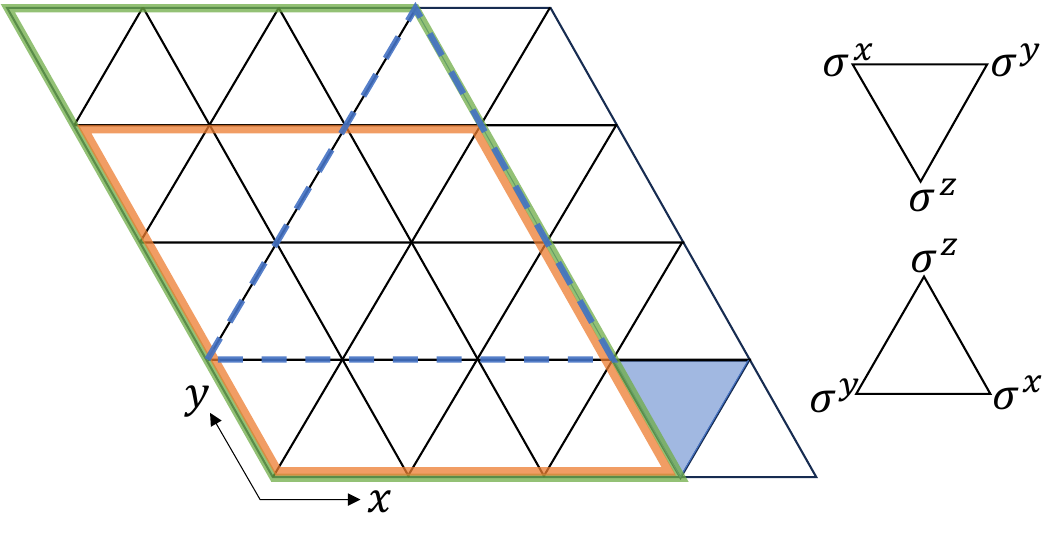}
    \caption{
    Construction used in the gauge-rank proof.
    The arrows define the lattice coordinates $(m,n)$ used in Eq.~\eqref{eq:SMABindex}.
    Products of triangles of the same type along a complete straight column reduce to the same double-loop operator, giving Eq.~\eqref{eq:SMcolumnidentity}.
    The dashed contour illustrates the triangular product in Eq.~\eqref{eq:SMseedtriangle}, whose interior factors cancel and leave the shaded triangle, providing the seed for constructing the final column.
    }
    \label{figS:static_rank}
\end{figure}

\subsection{Rank of the gauge group}
\label{sec:SMrankG}

We first prove the gauge-group rank appearing in Proposition~1 of the main text.

\begin{proposition}[Gauge-group basis]
\label{prop:SMgauge_rank}
For the $L\times L$ periodic Majorana XYZ code, the set
\begin{equation}
    \mathcal B=\{A_{m,n}:0\le m\le L-2,\ 0\le n\le L-1\}\cup\{B_{m,n}:0\le m,n\le L-2\}
\label{eq:SMgauge_basis}
\end{equation}
is an independent generating set of the gauge group $\mathcal G$.
Consequently,
\begin{equation}
\operatorname{rank}\mathcal G=|\mathcal B|=L(L-1)+(L-1)^2=(2L-1)(L-1).
\label{eq:SMrankG}
\end{equation}
\end{proposition}

\begin{proof}
We use the lattice convention shown in Fig.~\ref{figS:static_rank}.
Let $(m,n)\in\mathbb Z_L^2$ label the lattice sites, with $m$ increasing along the $x$ direction and $n$ along the $y$ direction indicated by the arrows.
With this convention,
\begin{equation}
    A_{m,n}=Y_{m,n}X_{m+1,n}Z_{m+1,n+1},\qquad B_{m,n}=Z_{m,n}X_{m,n+1}Y_{m+1,n+1},
\label{eq:SMABindex}
\end{equation}
where all indices are understood modulo $L$.

We first show that $\mathcal B$ generates $\mathcal G$.
For each fixed $m$, multiplying all triangles of either type in one column gives the same double-loop operator,
\begin{equation}
    \prod_{n=0}^{L-1}A_{m,n}=\prod_{n=0}^{L-1}B_{m,n}=\Xi_y^m\Xi_y^{m+1}.
\label{eq:SMcolumnidentity}
\end{equation}
Therefore, for $m=0,\ldots,L-2$, the missing $B$ triangle is generated by the other $L$ up-triangles and $L-1$ down-triangles in the same column.
Hence all $A$- and $B$-type triangles in the first $L-1$ columns are generated by $\mathcal B$.

It remains to generate the final column.
The triangular product indicated by the dashed contour in Fig.~\ref{figS:static_rank} gives
\begin{equation}
    B_{L-1,0}=\left[\prod_{m=1}^{L-2}\prod_{n=1}^{m}B_{m,n}\right]\left[\prod_{m=0}^{L-2}\prod_{n=1}^{m+1}A_{m,n}\right],
\label{eq:SMseedtriangle}
\end{equation}
up to an irrelevant overall phase.
All operators on the right-hand side have already been generated, so Eq.~\eqref{eq:SMseedtriangle} gives the first triangle in the final column.
The remaining triangles in that column follow successively by translation along the $y$ direction together with $\prod_{m=0}^{L-1}A_{m,n}=\prod_{m=0}^{L-1}B_{m,n}$.
Therefore $\mathcal B$ generates the full gauge group $\mathcal G$.

We next prove independence.
Suppose
\begin{equation}
    \prod_{m,n}A_{m,n}^{a_{m,n}}B_{m,n}^{b_{m,n}}=I,\qquad a_{m,n},b_{m,n}\in\mathbb F_2,
\label{eq:SMbasisrelation}
\end{equation}
is a relation involving only generators in $\mathcal B$.
Consider first the boundary site $(L-1,L)$.
Since $B_{L-2,L-1}$ is absent from $\mathcal B$, cancellation of the Pauli operators at this site requires $a_{L-2,L-1}=a_{L-2,0}=0$.
Moving successively along the same boundary $(L-1,n)$ then forces $a_{L-2,n}=b_{L-2,n}=0$ for all $n$, so the entire column $m=L-2$ is absent from the relation.
The same argument can then be applied successively to the newly exposed columns $m=L-3,L-4,\ldots,0$, yielding $a_{m,n}=b_{m,n}=0$ for every generator in $\mathcal B$.
Hence, $\mathcal B$ is independent.

Therefore, $\mathcal B$ is an independent generating set of $\mathcal G$, proving Eq.~\eqref{eq:SMrankG}.
\end{proof}

\subsection{Stabilizer group and parity dependence}
\label{sec:SMstabilizer}

We next determine the stabilizer subgroup appearing in Proposition~1 of the main text.

\begin{proposition}[Stabilizer group]
\label{prop:SMstabilizer}
The stabilizer group $\mathcal S=Z(\mathcal G)$ has
\begin{equation}
    \operatorname{rank}\mathcal S=3(L-1).
\label{eq:SMrankS}
\end{equation}
For odd $L$, it is generated by $L-1$ independent double loops of each Pauli type.
For even $L$, the double loops satisfy one additional relation and have rank $3(L-1)-1$; an additional independent stabilizer may be chosen as $H_L=\Xi_x^L\Xi_y^L\Xi_z^L$.
For odd $L$, $H_L$ is generated by the double loops, whereas for even $L$ it is independent of them.
\end{proposition}

\begin{proof}
The natural loops $\Xi_x^j$, $\Xi_y^j$, and $\Xi_z^j$ commute with every triangle gauge generator.
Moreover, neighboring parallel loops are generated by products of triangle operators over the strip between them.
For example, Eq.~\eqref{eq:SMcolumnidentity} gives $D_y^m=\Xi_y^m\Xi_y^{m+1}\in\mathcal G$.
Analogous products along the horizontal and diagonal directions give $D_z^j,D_x^j\in\mathcal G$.
Since these operators also commute with $\mathcal G$,
\begin{equation}
    D_a^j=\Xi_a^j\Xi_a^{j+1}\in Z(\mathcal G),
    \qquad a=x,y,z.
\label{eq:SMDinS}
\end{equation}
Within each family, $\prod_{j=1}^{L}D_a^j=I$, so at most $L-1$ double loops of each type are independent.

We now determine the relations between the three families.
Consider a product of elementary natural loops and denote by $x_j,y_j,z_j\in\mathbb F_2$ whether $\Xi_x^j,\Xi_y^j,\Xi_z^j$ occurs.
At a lattice site $(m,n)$, cancellation of the local Pauli operator requires the diagonal $x$ loop, vertical $y$ loop, and horizontal $z$ loop through that site to have the same occupation.
Hence all $x_j$, $y_j$, and $z_j$ must be equal.
Thus a product of natural loops is the identity only in two cases: either no loop is present, or all $L$ loops of all three families are present.

On the other hand, any product of double loops contains an even number of elementary loops within each family.
For odd $L$, the second possibility above contains an odd number $L$ of loops in every family and therefore cannot arise from double loops.
Consequently, after removing the single relation $\prod_jD_a^j=I$ within each family, the resulting $3(L-1)$ double loops are independent.

For even $L$, all $L$ elementary loops in a family can be generated by alternating double loops.
This gives one additional relation,
\begin{equation}
    \left(\prod_{j\ {\rm even}}D_x^j\right)
    \left(\prod_{j\ {\rm even}}D_y^j\right)
    \left(\prod_{j\ {\rm even}}D_z^j\right)=I,
\label{eq:SMevenDrelation}
\end{equation}
up to an irrelevant overall phase.
Since the only nontrivial identity among elementary loops is the one containing all three complete loop families, this is the only additional relation.
The double-loop subgroup therefore has rank $3(L-1)-1$ for even $L$.

The missing stabilizer can be chosen as
\begin{equation}
    H_L=\Xi_x^L\Xi_y^L\Xi_z^L.
\label{eq:SMHL}
\end{equation}
To verify that $H_L\in\mathcal G$, choose
\begin{equation}
    \Xi_x^L=\prod_{m=0}^{L-1}X_{m,m},\qquad
    \Xi_y^L=\prod_{n=0}^{L-1}Y_{0,n},\qquad
    \Xi_z^L=\prod_{m=0}^{L-1}Z_{m,0}.
\label{eq:SMHLloops}
\end{equation}
Using Eq.~\eqref{eq:SMABindex}, the product of all $A$-type triangles in the triangular region $0\le n\le m\le L-1$ gives
\begin{equation}
    \prod_{m=0}^{L-1}\prod_{n=0}^{m}A_{m,n}
    =
    \Xi_x^L\Xi_y^L\Xi_z^L
    =
    H_L,
\label{eq:SMHLgauge}
\end{equation}
up to an irrelevant overall phase.
All interior Pauli operators cancel, leaving only the three boundary loops; their common intersection contributes only an overall phase.
Hence $H_L\in\mathcal G$.
Since the three loop factors commute with every gauge generator, $H_L\in Z(\mathcal G)$.

The complete loop families satisfy
\begin{equation}
    \left(\prod_{j=1}^{L}\Xi_x^j\right)
    \left(\prod_{j=1}^{L}\Xi_y^j\right)
    \left(\prod_{j=1}^{L}\Xi_z^j\right)=I.
\label{eq:SMalllooprelation}
\end{equation}
Therefore $H_L$ is equivalent to the product of the remaining $L-1$ loops of each type.
For odd $L$, $L-1$ is even, and these loops can be paired into neighboring double loops; explicitly,
\begin{equation}
    H_L=
    \prod_{a=x,y,z}
    \prod_{\substack{j=1\\ j\ {\rm odd}}}^{L-2}D_a^j .
\label{eq:SMHLodd}
\end{equation}
Hence $H_L$ is generated by the double loops.
For even $L$, however, $L-1$ is odd, while every product of double loops contains an even number of elementary loops in each family.
Thus $H_L$ is independent of the double-loop subgroup and supplies the missing stabilizer generator.

It remains to show that the stabilizers constructed above exhaust $Z(\mathcal G)$.
Choose intersecting natural loops $\Xi_x$ and $\Xi_z$.
Both commute with every element of $\mathcal G$, but they anticommute with each other.
Therefore neither can belong to $\mathcal G$, and the subsystem code contains at least one logical qubit, $k\ge1$.
Using
\begin{equation}
    k=L^2-\frac{\operatorname{rank}\mathcal G+\operatorname{rank}\mathcal S}{2}
\label{eq:SMsubsystemk}
\end{equation}
together with $\operatorname{rank}\mathcal G=(2L-1)(L-1)$ gives $\operatorname{rank}\mathcal S\le3(L-1)$.
Since we have explicitly constructed $3(L-1)$ independent stabilizers for either parity of $L$, the bound is saturated and Eq.~\eqref{eq:SMrankS} follows.
\end{proof}

Combining Propositions~\ref{prop:SMgauge_rank} and \ref{prop:SMstabilizer} gives one logical qubit and
\begin{equation}
    r=\frac{\operatorname{rank}\mathcal G-\operatorname{rank}\mathcal S}{2}
    =(L-1)(L-2)
\label{eq:SMgaugequbits}
\end{equation}
gauge qubits.
A convenient logical pair is $\widetilde X=\Xi_x$ and $\widetilde Z=\Xi_z$.
They commute with $\mathcal G$, lie outside $\mathcal G$, and anticommute when chosen on intersecting noncontractible lines.

\subsection{Distance}
\label{sec:SMstaticdistance}

We finally determine the distance and complete the proof of Proposition~1 in the main text.

\begin{lemma}[Disjoint representatives]
\label{lem:SMdisjoint}
Let $Q$ be a logical Pauli operator with $M$ equivalent representatives
$Q_j=Qs_j$, where $s_j\in\mathcal S$, whose supports are pairwise disjoint.
Then any logical Pauli operator $P$ satisfying $\{P,Q\}=0$ has
$\operatorname{wt}(P)\ge M$.
\end{lemma}

\begin{proof}
Since $P$ commutes with every stabilizer,
$\{P,Q_j\}=0$ for all $j$.
Thus $P$ must overlap the support of every $Q_j$ on at least one qubit.
Because these supports are pairwise disjoint,
$\operatorname{wt}(P)\ge M$.
\end{proof}

\begin{proposition}[Static-code distance]
\label{prop:SMstaticdistance}
The static Majorana XYZ subsystem code has distance
\begin{equation}
    d=L.
\label{eq:SMstaticDistance}
\end{equation}
Consequently, its parameters are
\begin{equation}
    [[L^2,1,(L-1)(L-2),L]].
\label{eq:SMstaticParameters}
\end{equation}
\end{proposition}

\begin{proof}
Recall the logical pair
$\widetilde X=\Xi_x$ and $\widetilde Z=\Xi_z$.
Multiplying $\Xi_x^j$ by successive $x$-type double-loop stabilizers generates all translated representatives, $\Xi_x^j\sim\Xi_x^{j+1}$, and similarly for the $z$-type loops.
Hence both $\widetilde X$ and $\widetilde Z$ admit $L$ equivalent representatives supported on the $L$ mutually disjoint parallel noncontractible lines.

Since the code encodes a single logical qubit, every nontrivial logical Pauli is equivalent to one of
$\widetilde X$, $\widetilde Z$, or
$\widetilde Y=i\widetilde X\widetilde Z$.
The $L$ mutually disjoint representatives of $\widetilde X$ imply, by Lemma~\ref{lem:SMdisjoint}, that any logical operator anticommuting with $\widetilde X$ has weight at least $L$.
This applies to both $\widetilde Z$ and $\widetilde Y$.
Similarly, the $L$ mutually disjoint representatives of $\widetilde Z$ imply that both $\widetilde X$ and $\widetilde Y$ have weight at least $L$.
Hence every nontrivial logical Pauli has weight at least $L$, and therefore $d\ge L$.

Conversely, $\Xi_x$ and $\Xi_z$ themselves are nontrivial logical operators of weight $L$, so $d\le L$.
Therefore $d=L$.
Together with
$k=1$ and $r=(L-1)(L-2)$ obtained above, this gives
Eq.~\eqref{eq:SMstaticParameters}.
\end{proof}

\section{Parallelogram measurement algebra}
\label{sec:SMparallelogram}

A neighboring pair of triangle gauge generators produces a four-qubit parallelogram of a single Pauli type,
\begin{equation}
    P_{a,p}=\prod_{r\in p}\sigma_r^a,\qquad a=x,y,z .
\label{eq:SMP}
\end{equation}
For fixed $a$, all translated parallelograms commute.
We denote the corresponding measurement group by $\mathcal P_a$.

\subsection{Rank and generating sets}
\label{sec:SMPrank}

\begin{proposition}[Parallelogram measurement algebra]
\label{prop:SMparallelogram}
For each $a=x,y,z$, the measurement group $\mathcal P_a$ is Abelian and has
\begin{equation}
    \operatorname{rank}\mathcal P_a=(L-1)^2.
\label{eq:SMPrank}
\end{equation}
It admits an equivalent independent generating set consisting of $(L-1)(L-2)$ local parallelograms and the $L-1$ independent $a$-type double loops.

Furthermore, $\mathcal P_a$ can be extended to a maximal Abelian subgroup of $\mathcal G$ of rank $L^2-1$.
For odd $L$, the double loops of the other two Pauli types provide $2(L-1)$ additional independent generators.
For even $L$, they provide $2(L-1)-1$ additional independent generators, and $H_L$ supplies the remaining one.
\end{proposition}

\begin{proof}
It is sufficient to consider $\mathcal P_x$.
All $x$-type parallelograms mutually commute.
There are $L^2$ translated parallelograms, while the products along the $L$ complete rows and $L$ complete columns give $2L-1$ independent relations.
Hence
\begin{equation}
    \operatorname{rank}\mathcal P_x
    \le L^2-(2L-1)=(L-1)^2.
\label{eq:SMPupper}
\end{equation}

Conversely, the $(L-1)^2$ parallelograms in a contiguous $(L-1)\times(L-1)$ block generate all remaining parallelograms through the row and column relations.
They are also linearly independent.
Indeed, suppose a product of operators from this block is the identity.
At a corner qubit, only one parallelogram in the chosen block acts nontrivially, so the coefficient of that corner parallelogram must vanish.
Removing it exposes the next corner, and repeating this peeling argument along the boundary and then inward forces every parallelogram coefficient to vanish.
Thus the $(L-1)^2$ operators are linearly independent, which means the upper bound in Eq.~\eqref{eq:SMPupper} is saturated, and $\operatorname{rank}\mathcal P_x=(L-1)^2$.

We next show that these local parallelograms generate all diagonal double loops $D_x^j=\Xi_x^j\Xi_x^{j+1}$.
As shown in Fig.~\ref{figS:P_operator_set}, the product of the parallelograms enclosed by the blue dashed contour cancels in the interior and leaves only the two neighboring diagonal noncontractible loops,
\begin{equation}
    \prod_{p\in\mathcal R_j}P_{x,p}=\Xi_x^j\Xi_x^{j+1}=D_x^j.
\label{eq:SMPxDx}
\end{equation}
Translating this construction by one lattice spacing generates every $D_x^j$.
Since the translation is by one line, the construction applies equally to odd and even $L$.

We now construct the alternative basis.
The orange region in Fig.~\ref{figS:P_operator_set} contains
$(L-1)(L-2)$ local parallelograms.
Using the column relations, these operators first generate the remaining parallelograms in the same set of columns, giving $L(L-2)$ local parallelograms in total.
We then use the double-loop relations in Eq.~\eqref{eq:SMPxDx}.
For the blue dashed contour shown in Fig.~\ref{figS:P_operator_set}, all parallelograms entering the relation are already known except for the one labeled $1$.
Since the corresponding double loop $D_x^j$ is also included in the generating set, the relation determines the parallelogram labeled $1$.
Translating the dashed contour by one lattice spacing gives the next relation, in which the only unknown parallelogram is the one labeled $2$.
Repeating this procedure successively generates all remaining
parallelograms in that column.
Finally, the row relations generate the parallelograms in the last column.

Hence the $(L-1)(L-2)$ local parallelograms in the orange region together with the $L-1$ independent double loops
$D_x^1,\ldots,D_x^{L-1}$ generate all of $\mathcal P_x$.
Their total number is $(L-1)(L-2)+(L-1)=(L-1)^2$.
Since this equals $\operatorname{rank}\mathcal P_x$, they are linearly independent and therefore form an alternative basis of $\mathcal P_x$.
Thus the two generating descriptions of $\mathcal P_x$ are equivalent.

\begin{figure}[t]
    \centering
    \includegraphics[width=0.4\columnwidth]{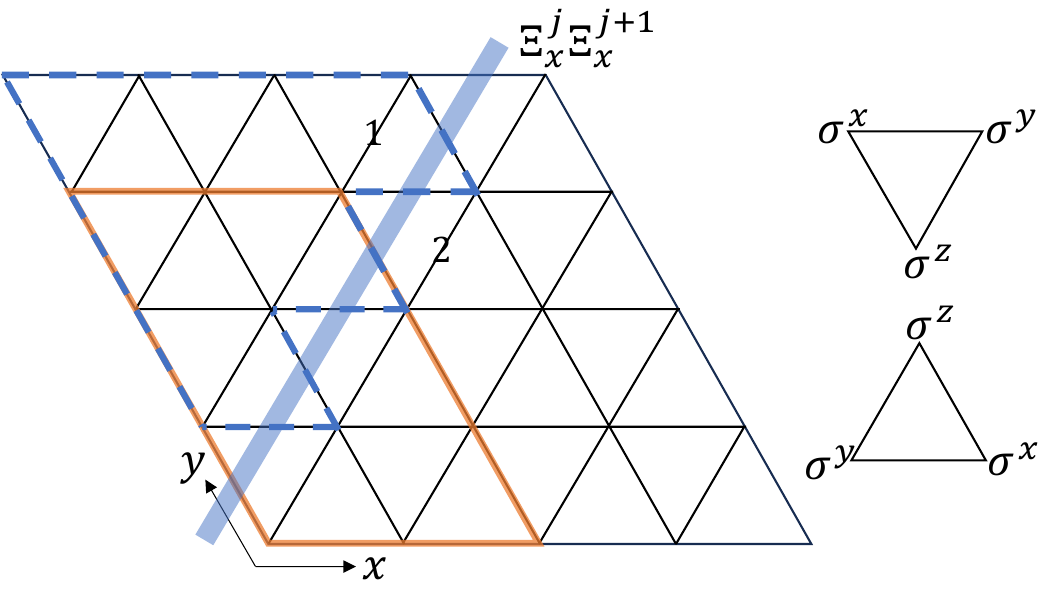}
    \caption{
    Construction of an equivalent generating set for $\mathcal P_x$.
    The orange region contains the $(L-1)(L-2)$ local $x$-type parallelograms retained in the alternative basis.
    The product of the parallelograms enclosed by the blue dashed contour cancels in the interior and leaves the neighboring noncontractible loops $\Xi_x^j$ and $\Xi_x^{j+1}$, giving the double loop $D_x^j=\Xi_x^j\Xi_x^{j+1}$.
    Translating this construction generates all $L-1$ independent $x$-type double loops.
    Together with the orange parallelograms, these double loops successively reconstruct the omitted local parallelograms, illustrated by labels $1$ and $2$, and hence generate the full group $\mathcal P_x$.}
    \label{figS:P_operator_set}
\end{figure}

We finally extend $\mathcal P_x$ by double loops of the other two Pauli types.
Since all double loops belong to the stabilizer group, they commute with every element of $\mathcal P_x$.
It remains to determine how many of the $y$- and $z$-type double loops are linearly independent modulo $\mathcal P_x$.

Consider a product
\begin{equation}
    Q_yQ_z,\qquad Q_y\in\mathcal D_y,\quad Q_z\in\mathcal D_z,
\label{eq:SMQyQz}
\end{equation}
and suppose that $Q_yQ_z\in\mathcal P_x$.
Since every element of $\mathcal P_x$ is a pure-$X$ Pauli operator, the $Y$ and $Z$ actions in $Q_yQ_z$ must combine to give only $I$ or $X$ on every qubit.
Let $y_j,z_k\in\mathbb F_2$ denote whether the elementary loops
$\Xi_y^j$ and $\Xi_z^k$ occur in $Q_y$ and $Q_z$, respectively.
At the intersection of $\Xi_y^j$ and $\Xi_z^k$, the operator is of pure $X$ type only if
$y_j=z_k$.
Since every $y$ loop intersects every $z$ loop, this condition for all intersections implies that all $y_j$ and $z_k$ are equal.
Thus there are only two possibilities: either no $y$ or $z$ loop occurs, or all $L$ loops of both families occur.

For odd $L$, every product of double loops contains an even number of elementary loops in each family, whereas the second possibility contains $L$ loops, which is odd.
It is therefore excluded.
Hence $Q_yQ_z\in\mathcal P_x$ implies $Q_y=Q_z=I$.
The $L-1$ independent $y$-type and $L-1$ independent $z$-type double loops are therefore all independent modulo $\mathcal P_x$ and contribute $2(L-1)$ additional generators.
The resulting commuting subgroup has rank
\begin{equation}
    (L-1)^2+2(L-1)=L^2-1.
\label{eq:SMmaxodd}
\end{equation}

For even $L$, the second possibility is allowed because $L$ is even.
The product of all $y$- and $z$-type loops is a pure-$X$ operator and is equal, up to phase, to the product of all $x$-type loops.
Equivalently, this is precisely the additional double-loop relation Eq.~\eqref{eq:SMevenDrelation}.
As shown in Proposition~\ref{prop:SMstabilizer}, this is the only additional relation among the three double-loop families.
Since the $x$-type double loops are already contained in $\mathcal P_x$, the $y$- and $z$-type families therefore contribute
$2(L-1)-1$ additional independent generators.

The remaining stabilizer $H_L$ is also independent of this enlarged group.
Indeed, suppose that
$H_LQ_yQ_z\in\mathcal P_x$ for some
$Q_y\in\mathcal D_y$ and $Q_z\in\mathcal D_z$.
The operator $H_L$ contributes one additional elementary $y$ loop and one additional elementary $z$ loop.
Consequently, $H_LQ_yQ_z$ contains an odd number of elementary loops in each of the $y$ and $z$ families.
For it to be a pure-$X$ operator, the same intersection argument above would require either none or all $L$ loops of each family.
Both possibilities contain an even number of loops for even $L$, giving a contradiction.
Thus
\begin{equation}
    H_L\notin
    \langle\mathcal P_x,\mathcal D_y,\mathcal D_z\rangle .
\label{eq:SMHLindependentPx}
\end{equation}
Adding $H_L$ therefore raises the rank to $(L-1)^2+\bigl[2(L-1)-1\bigr]+1=L^2-1$.

Finally, the quotient $\mathcal G/\mathcal S$ contains
$\operatorname{rank}\mathcal G-\operatorname{rank}\mathcal S$
noncentral Pauli generators, which occur in anticommuting pairs.
Hence a maximal Abelian subgroup of $\mathcal G$ has rank
\begin{equation}
    \frac{\operatorname{rank}\mathcal G+\operatorname{rank}\mathcal S}{2}
    =L^2-1.
\label{eq:SMmaxrank}
\end{equation}
The commuting subgroups constructed above attain this rank and are therefore maximal.
The results for $\mathcal P_y$ and $\mathcal P_z$ follow by cyclic symmetry.
\end{proof}

\subsection{Centralizer under successive measurements}
\label{sec:SMcentralizer}

The parity dependence of the Floquet protocol enters through the subgroup of one measurement family that commutes with the next.
For later convenience, define
\begin{equation}
    E_x^j=\Xi_{h,x}^j\Xi_{h,x}^{j+1},
\label{eq:SMEx}
\end{equation}
namely the nearest-neighbor pair of horizontal $x$ loops.

\begin{proposition}[Successive-measurement centralizer]
\label{prop:SMcentralizer}
For odd $L$,
\begin{equation}
    \mathcal P_x\cap C(\mathcal P_y)=\mathcal D_x .
\label{eq:SModdcentralizer}
\end{equation}
For even $L$,
\begin{equation}
    \mathcal P_x\cap C(\mathcal P_y)
    =
    \left\langle
    D_x^1,\ldots,D_x^{L-1},
    E_x^1,\ldots,E_x^{L-1}
    \right\rangle ,
\label{eq:SMevencentralizer1}
\end{equation}
where the two families on the right obey one additional relation and therefore have total rank $2L-3$.
After adjoining the newly measured family $\mathcal P_y$, this can equivalently be written as
\begin{equation}
    \left\langle\mathcal P_y,\,\mathcal P_x\cap C(\mathcal P_y)\right\rangle=\langle\mathcal P_y,\mathcal D_x,\mathcal D_z\rangle .
\label{eq:SMevencentralizer}
\end{equation}
The corresponding results for the other successive pairs follow by
cyclic permutation of $x,y,z$.
\end{proposition}

\begin{proof}
Because every element of $\mathcal P_x$ is a pure-$X$ Pauli operator, we represent it by a binary function $s_{u,v}\in\mathbb F_2$, where $s_{u,v}=1$ if the operator acts as $X$ on the qubit $(u,v)$ and $s_{u,v}=0$ otherwise.
Indices are understood modulo $L$.
We choose the coordinates so that increasing $u$ follows the horizontal lattice direction and increasing $v$ follows the diagonal direction shown in Fig.~\ref{figS:P_operator_set}.

Commutation with a $y$-type parallelogram requires an even number of overlapping $X$ operators, and therefore
\begin{equation}
    s_{u,v}+s_{u+1,v}+s_{u,v+1}+s_{u+1,v+1}=0\qquad (\mathrm{mod}\ 2).
\label{eq:SMbinary}
\end{equation}
For fixed $u$, this implies that $s_{u+1,v}+s_{u,v}=s_{u+1,v+1}+s_{u,v+1}$ is independent of $v$.
Comparing with $v=0$ gives $s_{u+1,v}+s_{u,v}=s_{u+1,0}+s_{u,0}$.
Hence the quantity $s_{u,v}+s_{u,0}=s_{u+1,v}+s_{u+1,0}$ is independent of $u$, so it equals $s_{0,v}+s_{0,0}$.
Therefore every solution can be written as
\begin{equation}
    s_{u,v}=r_u+c_v,
\label{eq:SMsolution}
\end{equation}
with, for example, $r_u=s_{u,0}$ and $c_v=s_{0,v}+s_{0,0}$.
Conversely, any configuration of this form satisfies Eq.~\eqref{eq:SMbinary}.
The decomposition has the redundancy $r_u\mapsto r_u+1$, $c_v\mapsto c_v+1$.

We must now impose the condition that the operator belongs not merely to the Pauli centralizer of $\mathcal P_y$, but specifically to $\mathcal P_x$.
Summing over $u$, we have
\begin{equation}
    0=\sum_u s_{u,v}=Lc_v+\sum_u r_u\qquad (\mathrm{mod}\ 2).
\label{eq:SMlineparity}
\end{equation}

For odd $L$, Eq.~\eqref{eq:SMlineparity} fixes $c_v=\sum_u r_u$, independently of $v$.
The $c_v$ contribution can therefore be absorbed into the redundant constant shift in Eq.~\eqref{eq:SMsolution}, leaving a pattern $s_{u,v}=r_u'$ that depends only on $u$.
Such a pattern is a product of complete diagonal $x$ loops $\Xi_x^j$.
Moreover, Eq.~\eqref{eq:SMlineparity} implies that an even number of these loops occurs.
Every even product of parallel loops is generated by nearest-neighbor double loops.
Hence
\begin{equation}
    \mathcal P_x\cap C(\mathcal P_y)=\mathcal D_x,
\end{equation}
which proves Eq.~\eqref{eq:SModdcentralizer}.

For even $L$, the term $Lc_v$ vanishes, so Eq.~\eqref{eq:SMlineparity} instead imposes $\sum_u r_u=0$ while leaving the $c_v$ variables unconstrained.
The periodic geometry supplies a second independent transverse-parity condition, which similarly gives $\sum_v c_v=0$.
Thus the $r_u$ and $c_v$ sectors both survive, but each contains an even number of complete lines.
The $r_u$ sector is generated by the diagonal double loops $D_x^j$, whereas the $c_v$ sector is generated by the horizontal double loops $E_x^j$, which can also be realized by the group $\mathcal{P}_x$ with even $L$.
This proves Eq.~\eqref{eq:SMevencentralizer1}.

For even $L$, the two families are not independent, because of the constraint
\begin{equation}
    \left(\prod_{j\ {\rm even}}D_x^j\right)\left(\prod_{j\ {\rm even}}E_x^j\right)=I,
\label{eq:SMDErelation}
\end{equation}
up to an irrelevant phase.
This is the only relation connecting the two families, and therefore
$\operatorname{rank}[\mathcal P_x\cap C(\mathcal P_y)]=2L-3$.

It remains to obtain the form used in the ISG.
Define similarly the horizontal $y$-loop pair
\begin{equation}
    E_y^j=\Xi_{h,y}^j\Xi_{h,y}^{j+1}.
\label{eq:SMEy}
\end{equation}
The product of the even-site $P_y$ parallelograms in the strip between two neighboring horizontal lines gives precisely $E_y^j$.
Therefore, we have $E_y^j\in\mathcal P_y$.
On the same two horizontal lines, $E_x^jE_y^j=\Xi_{h,z}^j\Xi_{h,z}^{j+1}=D_z^j$
up to an irrelevant Pauli phase.
Consequently, $E_x^j$ and $D_z^j$ differ by an element of $\mathcal P_y$.
Replacing every $E_x^j$ by $D_z^j$ after adjoining $\mathcal P_y$ therefore gives
\begin{equation}
    \left\langle\mathcal P_y,\,\mathcal P_x\cap C(\mathcal P_y)\right\rangle=\langle\mathcal P_y,\mathcal D_x,\mathcal D_z\rangle ,
\end{equation}
which proves Eq.~\eqref{eq:SMevencentralizer}.

The other successive-measurement centralizers follow by cyclic permutation of the lattice directions and Pauli labels.
\end{proof}

\section{Three-step protocol}
\label{sec:SMT3}

We now prove Theorem~1 of the main text.
The measurement sequence is
\begin{equation}
    \mathcal P_x\longrightarrow
    \mathcal P_y\longrightarrow
    \mathcal P_z\longrightarrow
    \mathcal P_x\longrightarrow\cdots ,
\label{eq:SMT3}
\end{equation}
with period $T=3$.
Up to signs determined by the measurement outcomes, the ISG updates as $\mathcal S_t=\langle\mathcal P_{a_t}, \mathcal S_{t-1}\cap C(\mathcal P_{a_t})\rangle$.

\begin{thm}[Three-step protocol]
\label{thm:SMT3}
After the initial transient, the three-step protocol has distance $d=L$.
For odd $L$, it encodes one static logical qubit.
For $L=4n+2$, it encodes two static logical qubits.
For $L=4n$, it encodes one static and one partially dynamical logical qubit.
\end{thm}

\subsection{Odd $L$}
\label{sec:SMT3odd}

We first determine the ISG evolution.
Starting from a trivial ISG, the first measurement gives $\mathcal S_0=\mathcal P_x$.
For odd $L$, Proposition~\ref{prop:SMcentralizer} gives $\mathcal P_x\cap C(\mathcal P_y)=\mathcal D_x$, and therefore
\begin{equation}
    \mathcal S_1=\langle\mathcal P_y,\mathcal D_x\rangle .
\label{eq:SModdS1}
\end{equation}
Since $\mathcal P_y$ has rank $(L-1)^2$ and already contains the $L-1$ independent $y$-type double loops, while $\mathcal D_x$ is independent of $\mathcal P_y$ for odd $L$, we obtain $\operatorname{rank}\mathcal S_1=L^2-L$.

At the next round, all elements of $\mathcal D_x$ remain in the ISG because the double loops belong to $Z(\mathcal G)$ and hence commute with $\mathcal P_z$.
Proposition~\ref{prop:SMcentralizer} also gives $\mathcal P_y\cap C(\mathcal P_z)=\mathcal D_y$.
Applying the ISG update rule therefore yields
\begin{equation}
    \mathcal S_2=
    \langle\mathcal P_z,\mathcal D_x,\mathcal D_y\rangle .
\label{eq:SModdS2}
\end{equation}
By Proposition~\ref{prop:SMparallelogram}, for odd $L$ the group $\mathcal P_z$ together with the $x$- and $y$-type double loops has rank $L^2-1$.
Hence $\operatorname{rank}\mathcal S_2=L^2-1$.

The following $\mathcal P_x$ round gives $\mathcal S_3=\langle\mathcal P_x,\mathcal D_y,\mathcal D_z\rangle$, and the subsequent ISGs are obtained by cyclic permutation of $x,y,z$.
Thus every steady-round ISG has rank $L^2-1$ and therefore encodes one logical qubit.

A convenient logical Pauli pair is
\begin{equation}
    (\widetilde X_1,\widetilde Z_1)=(\Xi_x,\Xi_z).
\label{eq:SModdlogical}
\end{equation}
Both operators belong to the centralizer of the static gauge group $\mathcal G$ and therefore commute with every measurement family $\mathcal P_a\subset\mathcal G$.
Since they anticommute with one another, they provide a time-independent logical Pauli pair throughout the entire cycle.
This encoded qubit is inherited directly from the static subsystem code.
The distance $d=L$ will be proved in Sec.~\ref{sec:SMT3distance}.

\subsection{Even $L$: common ISG structure}
\label{sec:SMT3even}

We next consider even $L$.
The first measurement again gives $\mathcal S_0=\mathcal P_x$.
At the second round, Proposition~\ref{prop:SMcentralizer} gives
\begin{equation}
    \mathcal S_1=\langle\mathcal P_y,\mathcal D_x,\mathcal D_z\rangle .
\label{eq:SMevenS1}
\end{equation}
By Proposition~\ref{prop:SMparallelogram}, this group has rank $L^2-2$.
Thus, the two-logical-qubit structure is already reached after the first two measurements.

Upon measuring $\mathcal P_z$, the double loops already present in $\mathcal S_1$ remain because they belong to $Z(\mathcal G)$.
Applying Proposition~\ref{prop:SMcentralizer} to the successive $\mathcal P_y\rightarrow\mathcal P_z$ measurements gives
\begin{equation}
    \mathcal S_2=\langle\mathcal P_z,\mathcal D_x,\mathcal D_y\rangle .
\label{eq:SMevenS2}
\end{equation}
The following $\mathcal P_x$ round gives $\mathcal S_3=\langle\mathcal P_x,\mathcal D_y,\mathcal D_z\rangle$, and the subsequent ISGs are obtained by cyclic permutation of $x,y,z$.
By Proposition~\ref{prop:SMparallelogram}, every steady-round ISG has rank $L^2-2$ and therefore encodes two logical qubits.

One of these logical qubits is again inherited directly from the static subsystem code.
We choose
\begin{equation}
    (\widetilde X_1,\widetilde Z_1)=(\Xi_x,\Xi_z).
\label{eq:SMevenlogical1}
\end{equation}
As in the odd-$L$ case, both operators commute with the full gauge group and hence with every measurement family, so this logical qubit remains static throughout the cycle.

The second logical qubit originates from the even-$L$ stabilizer $H_L$ of the static subsystem code.
Since $H_L\in Z(\mathcal G)$, it commutes with every steady ISG.
On the other hand, Proposition~\ref{prop:SMparallelogram} shows that adjoining $H_L$ to $\langle\mathcal P_a,\mathcal D_b,\mathcal D_c\rangle$ raises its rank from $L^2-2$ to $L^2-1$.
Therefore $H_L$ is not contained in the steady ISG and represents a nontrivial logical Pauli operator.
We choose
\begin{equation}
    \widetilde Z_2=H_L.
\label{eq:SMevenZ2}
\end{equation}
The existence of a time-independent conjugate $\widetilde X_2$ depends on $L\bmod 4$.

\subsection{$L=4n+2$: two static logical qubits}
\label{sec:SM4n2}

We first consider $L=4n+2$.
Define
\begin{equation}
    \Omega_x=\prod_{j\ {\rm even}}\Xi_{h,x}^j .
\label{eq:SMOmegax}
\end{equation}
Every parallelogram has an even number of anticommuting overlaps with $\Omega_x$, and hence $[\Omega_x,\mathcal P_a]=0$ for $a=x,y,z$.
Thus $\Omega_x$ commutes with every ISG throughout the three-step cycle.

A single horizontal $x$ loop has odd intersection parity with $H_L$.
Since $\Omega_x$ contains $L/2$ such loops, their commutation relation is $\Omega_xH_L=(-1)^{L/2}H_L\Omega_x$.
For $L=4n+2$, $L/2=2n+1$ is odd, and therefore
\begin{equation}
    \{\Omega_x,H_L\}=0 .
\label{eq:SMOmegaanti}
\end{equation}
Moreover, $\Omega_x$ commutes with the first logical pair $(\widetilde X_1,\widetilde Z_1)$.
Hence
\begin{equation}
    (\widetilde X_2,\widetilde Z_2)=(\Omega_x,H_L)
\label{eq:SM4n2pair}
\end{equation}
defines a second time-independent logical Pauli pair.

Therefore, for $L=4n+2$, both encoded qubits admit static logical representatives throughout the three-step cycle.

\begin{figure}[t]
    \centering
    \includegraphics[width=0.75\columnwidth]{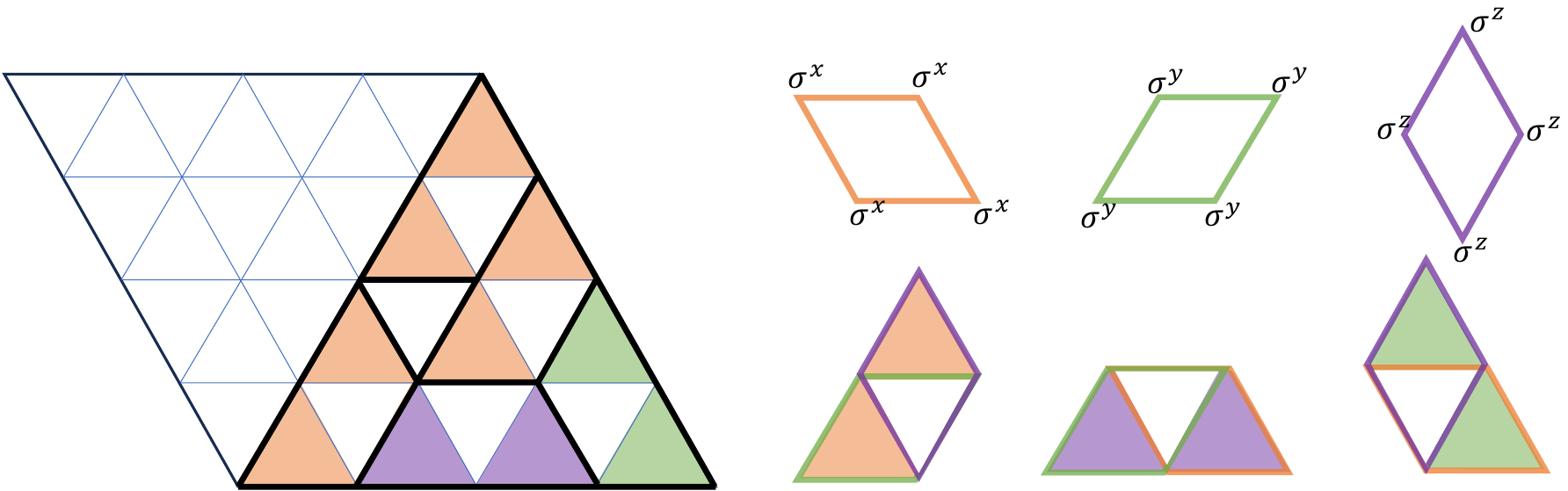}
    \caption{
    Decomposition used to express $H_L$ in terms of the three measurement families for $L=4n$.
    Left: the triangular region of up-triangle gauge generators whose product gives $H_L$.
    The up triangles are partitioned into the three types of trapezoidal blocks indicated by different colors.
    Right: the three local identities obtained from the products $P_xP_y$, $P_yP_z$, and $P_zP_x$.
    Each product gives a trapezoid containing two neighboring up-triangle generators.
    Multiplying these identities over all trapezoids in the left panel reproduces the product of all up triangles, and hence expresses $H_L$ as an element of $\langle\mathcal P_x,\mathcal P_y,\mathcal P_z\rangle$.}
    \label{figS:expressHL}
\end{figure}

\subsection{$L=4n$: obstruction to a static conjugate}
\label{sec:SM4n}

For $L=4n$, the operator $\Omega_x$ introduced above contains $L/2=2n$ horizontal $x$ loops and therefore commutes with $H_L$.
Thus it cannot serve as a conjugate of $\widetilde Z_2=H_L$.
We now show that this obstruction is unavoidable.

Let
\begin{equation}
    \mathcal M_{\rm cyc}=\langle\mathcal P_x,\mathcal P_y,\mathcal P_z\rangle
\label{eq:SMMcyc}
\end{equation}
be the group generated by all measurements appearing in one period.

\begin{lemma}
\label{lem:SMHLinM}
For $L=4n$, $H_L\in\mathcal M_{\rm cyc}$.
\end{lemma}

\begin{proof}
Recall from Sec.~\ref{sec:SMstabilizer} that $H_L$ is generated by the product of all up-triangle gauge generators in the triangular region shown in the left panel of Fig.~\ref{figS:expressHL}.

The three local identities shown on the right of Fig.~\ref{figS:expressHL} are obtained by multiplying neighboring parallelograms of two different Pauli types.
Specifically, the products $P_xP_y$, $P_yP_z$, and $P_zP_x$ give the three possible trapezoidal blocks shown in the figure.
Each trapezoid is equal, up to an irrelevant Pauli phase, to the product of the two up-triangle gauge generators contained in that trapezoid.

For $L=4n$, the complete set of up triangles in the triangular region can be partitioned into these trapezoidal blocks, as indicated by the three colors in the left panel.
Multiplying all of the corresponding local identities therefore gives, on one side, the product of all up-triangle generators in the region, which is $H_L$, and, on the other side, a product solely of parallelogram operators from $\mathcal P_x$, $\mathcal P_y$, and $\mathcal P_z$.
Hence
\begin{equation}
    H_L\in\langle\mathcal P_x,\mathcal P_y,\mathcal P_z\rangle=\mathcal M_{\rm cyc}.
\label{eq:SMHLinM}
\end{equation}
\end{proof}

\begin{thm}[Obstruction to a static conjugate]
\label{thm:SMstaticobstruction}
For $L=4n$, no time-independent Pauli operator can both commute with every measurement in the three-step cycle and anticommute with $H_L$.
\end{thm}

\begin{proof}
A time-independent logical representative must commute with every measurement operator and hence with every element of $\mathcal M_{\rm cyc}$.
By Lemma~\ref{lem:SMHLinM}, $H_L\in\mathcal M_{\rm cyc}$, so any such operator must commute with $H_L$.
This contradicts the anticommutation relation required of a logical Pauli conjugate to $\widetilde Z_2=H_L$.
\end{proof}

Thus the second logical qubit cannot admit a time-independent Pauli pair for $L=4n$.
Since $\widetilde Z_2=H_L$ remains static, its conjugate must change during the measurement cycle.

\subsection{Explicit dynamical logical operator}
\label{sec:SMdynamic}

We now construct an explicit conjugate of $\widetilde Z_2=H_L$ at each measurement round.
For $a=x,y,z$, define
\begin{equation}
    F_a=\Xi_{h,a}\Xi_{v,a}\Xi_{d,a}.
\label{eq:SMFa}
\end{equation}
Figure~\ref{figS:transferF} illustrates the construction for $a=y$.
Multiplying the shaded $y$-type parallelograms in the triangular region cancels all interior Pauli operators pairwise and leaves only the three noncontractible boundary loops $\Xi_{h,y}$, $\Xi_{v,y}$, and $\Xi_{d,y}$, namely $F_y$.
The same construction applies cyclically to $a=x$ and $a=z$.
Hence $F_a\in\mathcal P_a$.

\begin{figure}[t]
    \centering
    \includegraphics[width=0.55\columnwidth]{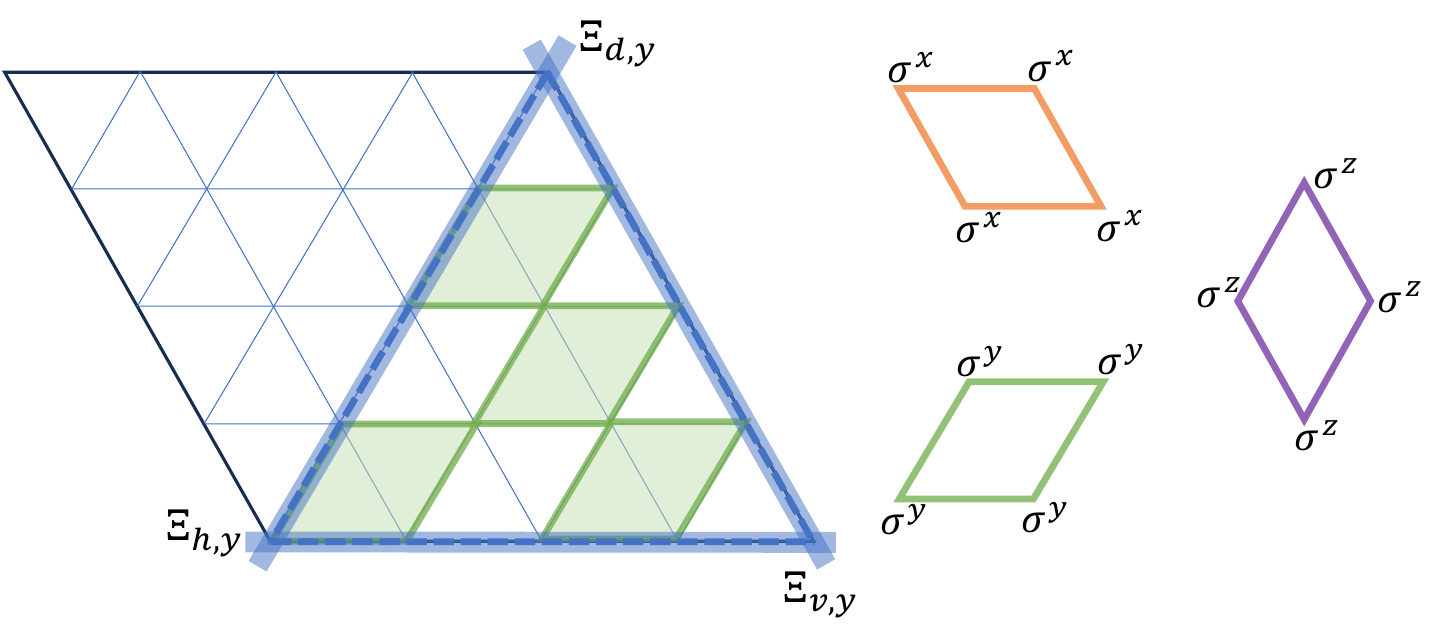}
    \caption{
    Construction of $F_a=\Xi_{h,a}\Xi_{v,a}\Xi_{d,a}$ from the measurement family $\mathcal P_a$.
    The figure shows the case $a=y$.
    Multiplying the shaded $y$-type parallelograms in the triangular region cancels all interior Pauli operators pairwise and leaves only the three noncontractible boundary loops, giving $F_y$.
    The constructions of $F_x$ and $F_z$ are obtained analogously by cyclic permutation.
    }
    \label{figS:transferF}
\end{figure}

A convenient sequence of representatives is
\begin{align}
    \widetilde X_2^{(0)}&=\Xi_{h,x},\nonumber\\
    \widetilde X_2^{(1)}&=\Xi_{h,z}\Xi_{v,y}\Xi_{d,y},\nonumber\\
    \widetilde X_2^{(2)}&=\Xi_{v,x}\Xi_{d,x}.
\label{eq:SMDynamicX}
\end{align}
To make the measurement times explicit, we write their evolution as
\begin{equation}
    \widetilde X_2^{(0)}
    \xrightarrow[\times F_y]{\mathcal P_y}
    \widetilde X_2^{(1)}
    \xrightarrow[\times F_z]{\mathcal P_z}
    \widetilde X_2^{(2)}
    \xrightarrow[\times F_x]{\mathcal P_x}
    \widetilde X_2^{(0)} .
\label{eq:SMDynamicOrbit}
\end{equation}
Each arrow denotes the corresponding measurement round, while the label below the arrow gives the measured operator relating the two representatives.

The representatives on the two sides of each arrow are both valid immediately after that measurement.
For example, $\widetilde X_2^{(0)}$ is valid before the $\mathcal P_y$ measurement and remains valid after it.
Moreover, $\widetilde X_2^{(1)}=F_y\widetilde X_2^{(0)}$ up to an irrelevant Pauli phase.
Since $F_y\in\mathcal P_y$ is fixed by the new measurement record, $\widetilde X_2^{(0)}$ and $\widetilde X_2^{(1)}$ are equivalent representatives of the same logical Pauli operator immediately after the $\mathcal P_y$ measurement.
We may then use $\widetilde X_2^{(1)}$ as the representative through the following $\mathcal P_z$ measurement.
The same argument applies cyclically to the other two arrows.

At each round, the corresponding representative commutes with the instantaneous stabilizer group and anticommutes with the fixed logical operator $H_L$,
\begin{equation}
    [\widetilde X_2^{(t)},\mathcal S_t]=0,\qquad\{\widetilde X_2^{(t)},H_L\}=0.
\label{eq:SMDynamicCommutation}
\end{equation}
Therefore, it represents a valid conjugate logical Pauli operator at that round.
Moreover, each update multiplies the representative by $F_a\in\mathcal P_a$, where $\mathcal P_a$ is the family measured in the new round.
Since the eigenvalue of $F_a$ is fixed by the corresponding measurement record, the logical information is preserved while its physical Pauli representative evolves through the three-step cycle.

Together with Theorem~\ref{thm:SMstaticobstruction}, this indicates that the second logical qubit for $L=4n$ is partially dynamical: $\widetilde Z_2=H_L$ admits a time-independent representative, whereas no time-independent representative exists for its conjugate.

It is also useful to track the third Pauli operator of the second logical qubit. 
Define
\begin{equation}
\widetilde Y_2^{(t)}=i\widetilde X_2^{(t)}\widetilde Z_2,
\qquad
\widetilde Z_2=H_L .
\end{equation}
Since $H_L$ is time-independent, $\widetilde Y_2^{(t)}$ obeys the same update rule as $\widetilde X_2^{(t)}$. 
Up to irrelevant Pauli phases, a convenient choice is
\begin{align}
\widetilde Y_2^{(0)}&=\Xi_{h,y}\Xi_{v,y}\Xi_{d,x},\nonumber\\
\widetilde Y_2^{(1)}&=\Xi_{d,z},\nonumber\\
\widetilde Y_2^{(2)}&=\Xi_{h,z}\Xi_{v,z}.
\end{align}
These representatives satisfy
\begin{equation}
\widetilde Y_2^{(0)}\xrightarrow[\times F_y]{\mathcal P_y}
\widetilde Y_2^{(1)}\xrightarrow[\times F_z]{\mathcal P_z}
\widetilde Y_2^{(2)}\xrightarrow[\times F_x]{\mathcal P_x}
\widetilde Y_2^{(0)} .
\end{equation}
Thus the continuously tracked logical frame may be chosen as $\widetilde X_1,\widetilde Z_1,\widetilde X_2^{(t)},\widetilde Z_2$.
The first logical qubit and $\widetilde Z_2$ remain fixed, while both $\widetilde X_2^{(t)}$ and $\widetilde Y_2^{(t)}$ are updated by multiplication with an operator generated by the newly measured family. 
Because the eigenvalue of each $F_a\in\mathcal P_a$ is fixed by the corresponding measurement record, the representatives on the two sides of each update carry the same logical Pauli information.

\subsection{Round-adapted logical basis and distance}
\label{sec:SMT3distance}

For the distance proof, it is convenient to use weight-$L$ logical representatives adapted to the instantaneous ISG.
For either even size class, we choose
\begin{center}
\begin{tabular}{c c c}
\toprule
measured family &
$(\widehat X_1,\widehat Z_1)$ &
$(\widehat X_2,\widehat Z_2)$\\
\midrule
$\mathcal P_x$ &
$(\Xi_{v,y},\Xi_{h,z})$ &
$(\Xi_{v,z},\Xi_{h,y})$\\
$\mathcal P_y$ &
$(\Xi_{h,z},\Xi_{d,x})$ &
$(\Xi_{h,x},\Xi_{d,z})$\\
$\mathcal P_z$ &
$(\Xi_{d,x},\Xi_{v,y})$ &
$(\Xi_{d,y},\Xi_{v,x})$\\
\bottomrule
\end{tabular}
\end{center}
The line operators in the table are chosen independently at each measurement round. 
To distinguish this round-adapted basis from the continuously tracked logical frame above, we denote them by $(\widehat X_1,\widehat Z_1)$ and $(\widehat X_2,\widehat Z_2)$. Their relation to the continuously tracked operators can be written explicitly.

At the $\mathcal P_x$ round,
\begin{equation}
    \widehat X_1=\Xi_{v,y}=\widetilde X_1\widetilde Z_1\widetilde Z_2,\qquad \widehat Z_1=\Xi_{h,z}=\widetilde Z_1,\qquad \widehat X_2=\Xi_{v,z}= F_x\widetilde Z_1\widetilde Y_2^{(0)},\qquad \widehat Z_2=\Xi_{h,y}\simeq\widetilde Z_1\widetilde X_2^{(0)} .
\end{equation}
At the $\mathcal P_y$ round,
\begin{equation}
\widehat X_1=\Xi_{h,z}=\widetilde Z_1,\qquad
\widehat Z_1=\Xi_{d,x}=\widetilde X_1,\qquad
\widehat X_2=\Xi_{h,x}= F_y\widetilde X_2^{(1)},\qquad
\widehat Z_2=\Xi_{d,z}=\widetilde Y_2^{(1)} .
\end{equation}
At the $\mathcal P_z$ round,
\begin{equation}
\widehat X_1=\Xi_{d,x}=\widetilde X_1,\qquad
\widehat Z_1=\Xi_{v,y}=\widetilde X_1\widetilde Z_1\widetilde Z_2,\qquad
\widehat X_2=\Xi_{d,y}= F_z\widetilde X_1\widetilde Y_2^{(2)},\qquad
\widehat Z_2=\Xi_{v,x}=\widetilde X_1\widetilde X_2^{(2)} .
\end{equation}
Equivalently, the continuously tracked second logical qubit can be reconstructed directly from the round-adapted basis. 
Since $F_a\in\mathcal P_a$ is generated by the measurement family at the corresponding round, its eigenvalue is known from the measurement record. 
These identities give an explicit conversion between the round-adapted weight-$L$ basis and the continuously tracked logical frame.
The individual round-adapted line operators therefore need not themselves follow a fixed logical Pauli class through the cycle. 
Their logical frame can change with the round and mix the two encoded qubits. 
This does not affect the distance argument, which requires only a complete instantaneous logical basis with sufficiently many disjoint equivalent representatives.

We now verify that the line operators appearing in the table admit the required translated representatives.
For this purpose, define the orientation-resolved neighboring-loop pair
\begin{equation}
    D_{\alpha,a}^j=\Xi_{\alpha,a}^j\Xi_{\alpha,a}^{j+1},
    \qquad \alpha=h,v,d .
\label{eq:SMgeneralD}
\end{equation}
The double loops introduced previously are $D_x^j=D_{d,x}^j$, $D_y^j=D_{v,y}^j$, and $D_z^j=D_{h,z}^j$.

For even $L$, the additional orientation-resolved double loops needed below are also generated by the corresponding measurement family.
For each orientation, the product of all even-position $P_a$ parallelograms along the corresponding translated sequence gives the neighboring-loop pair $D_{\alpha,a}^j$.
In particular, $D_{v,x}^j,D_{h,x}^j\in\mathcal P_x$, $D_{h,y}^j,D_{d,y}^j\in\mathcal P_y$, and $D_{d,z}^j,D_{v,z}^j\in\mathcal P_z$.

The remaining double loops follow by combining these operators with the natural double-loop stabilizers already contained in the corresponding ISG.
Up to irrelevant Pauli phases,
\begin{align}
\mathcal P_x:\quad&
D_{v,z}^j=D_{v,x}^jD_y^j,\qquad
D_{h,y}^j=D_{h,x}^jD_z^j,\nonumber\\
\mathcal P_y:\quad&
D_{h,x}^j=D_{h,y}^jD_z^j,\qquad
D_{d,z}^j=D_{d,y}^jD_x^j,\nonumber\\
\mathcal P_z:\quad&
D_{d,y}^j=D_{d,z}^jD_x^j,\qquad
D_{v,x}^j=D_{v,z}^jD_y^j .
\label{eq:SMrounddoubleloops}
\end{align}
Therefore the instantaneous ISG contains the neighboring-loop pair associated with every line operator appearing in the corresponding row of the table.
Multiplying a line operator successively by these double-loop stabilizers translates it through all $L$ parallel positions and produces $L$ mutually disjoint equivalent representatives.

We can now prove the distance.
Let $P$ be any nontrivial logical Pauli operator at a steady round.
Because the four operators in the corresponding row form a complete symplectic basis for the two encoded qubits, $P$ anticommutes with at least one of them.
That logical generator has $L$ mutually disjoint equivalent representatives, so Lemma~\ref{lem:SMdisjoint} gives $\operatorname{wt}(P)\ge L$.
Therefore $d\ge L$.
Conversely, the table contains nontrivial logical representatives of weight $L$, so $d\le L$.
Hence
\begin{equation}
    d=L.
\label{eq:SMT3distance}
\end{equation}

For odd $L$, the same argument applies to the single logical qubit.
The steady ISG contains the appropriate natural double-loop stabilizers generating $L$ mutually disjoint representatives of $\widetilde X_1=\Xi_x$ and $\widetilde Z_1=\Xi_z$.
Therefore the odd-$L$ three-step code also has distance $d=L$.

This completes the proof of Theorem~\ref{thm:SMT3}.

\section{Two-step protocol for even $L$}
\label{sec:SMT2}

We next consider the two-step measurement cycle
\begin{equation}
    \mathcal P_x\longrightarrow\mathcal P_y \longrightarrow\mathcal P_x\longrightarrow\cdots ,
\label{eq:SMT2}
\end{equation}
with period $T=2$.

\begin{proposition}[Two-step protocol]
\label{prop:SMT2}
For even $L$, the two-step protocol encodes two static logical qubits with distance $d=L$ after the initial transient.
\end{proposition}

\begin{proof}
Starting from a trivial ISG, the first measurement gives $\mathcal S_0=\mathcal P_x$.
Using Proposition~\ref{prop:SMcentralizer}, the subsequent $\mathcal P_y$ measurement gives $\mathcal S_1=\langle\mathcal P_y,\mathcal D_x,\mathcal D_z\rangle$.
The following $\mathcal P_x$ measurement similarly gives $\mathcal S_2=\langle\mathcal P_x,\mathcal D_y,\mathcal D_z\rangle$.
The two ISGs then alternate periodically.
By Proposition~\ref{prop:SMparallelogram}, both have rank $L^2-2$ and therefore encode two logical qubits.

The first logical qubit is again inherited from the static subsystem code and may be chosen as
\begin{equation}
    (\widetilde X_1,\widetilde Z_1)=(\Xi_x,\Xi_z).
\label{eq:SMT2logical1}
\end{equation}
Both operators commute with the full gauge group and hence remain time independent throughout the two-step cycle.

For the second logical qubit, we again choose
\begin{equation}
    \widetilde Z_2=H_L.
\label{eq:SMT2Z2}
\end{equation}
Since $H_L\in Z(\mathcal G)$, it commutes with both steady ISGs, while Proposition~\ref{prop:SMparallelogram} implies that it is not generated by either of them.
Thus $H_L$ represents a nontrivial logical Pauli operator.

For $L=4n+2$, the static conjugate may be chosen as $\widetilde X_2=\Omega_x$.
As shown in Sec.~\ref{sec:SM4n2}, $\Omega_x$ commutes with all three parallelogram families and anticommutes with $H_L$, so it is in particular a valid time-independent conjugate for the two-step cycle.

For $L=4n$, the absence of the $\mathcal P_z$ measurement removes the obstruction found in Sec.~\ref{sec:SM4n}.
A convenient choice is
\begin{equation}
    \widetilde X_2=\Xi_{h,x}.
\label{eq:SMT2X24n}
\end{equation}
Direct intersection counting shows that $\Xi_{h,x}$ commutes with both $\mathcal P_x$ and $\mathcal P_y$, while $\{\Xi_{h,x},H_L\}=0$.
Hence $(\Xi_{h,x},H_L)$ is a time-independent logical Pauli pair.

Equivalently, the operators $\Omega_x$ for $L=4n+2$ and $\Xi_{h,x}$ for $L=4n$ commute with every measurement appearing in the two-step cycle while anticommuting with $H_L$.
Therefore $H_L\notin\langle\mathcal P_x,\mathcal P_y\rangle$ in either even-$L$ size class.
This contrasts with the three-step $L=4n$ protocol, where $H_L\in\langle\mathcal P_x,\mathcal P_y,\mathcal P_z\rangle$ and no static conjugate can exist.

Moreover, the simple weight-$L$ representatives introduced in Sec.~\ref{sec:SMT3distance} can be chosen to remain valid throughout the entire two-step cycle.
In particular, the $\mathcal P_x$-row representatives
\begin{equation}
    (\widetilde X_1,\widetilde Z_1)=(\Xi_{v,y},\Xi_{h,z}),\qquad
    (\widetilde X_2,\widetilde Z_2)=(\Xi_{v,z},\Xi_{h,y})
\label{eq:SMT2simplebasis}
\end{equation}
commute with both steady ISGs $\langle\mathcal P_x,\mathcal D_y,\mathcal D_z\rangle$ and $\langle\mathcal P_y,\mathcal D_x,\mathcal D_z\rangle$.
Thus, unlike in the three-step $L=4n$ protocol, no change of physical representative is required during the two-step cycle.
By the logical-class identification established in Sec.~\ref{sec:SMT3distance}, these operators are equivalent, modulo the corresponding ISG, to the static logical operators defined above.

Each line operator in Eq.~\eqref{eq:SMT2simplebasis} has $L$ mutually disjoint equivalent representatives generated by the corresponding double-loop stabilizers.
The disjoint-representative argument of Sec.~\ref{sec:SMT3distance} therefore applies directly and gives $d=L$ throughout the steady two-step cycle.
\end{proof}

\section{One-step rectangle stabilizer code}
\label{sec:SMT1}

We finally consider the one-step protocol discussed in the main text.
Instead of the parallelogram families, we measure the commuting four-qubit rectangle operators
\begin{equation}
    R_p=\sigma_i^x\sigma_j^y\sigma_k^x\sigma_l^y,
\label{eq:SMrectangle}
\end{equation}
with the geometry shown in Fig.~\ref{figS:T1logical}.
We denote by $\mathcal R=\langle R_p\rangle$ the stabilizer group generated by all translated rectangles.
Since the same commuting family is measured at every round, the Floquet period is $T=1$ and the instantaneous stabilizer group is simply $\mathcal R$.

\begin{proposition}[One-step rectangle code]
\label{prop:SMT1}
For odd $L$, the rectangle code has $\operatorname{rank}\mathcal R=L^2-1$ and encodes one static logical qubit with distance $d=L$.
For even $L$, it has $\operatorname{rank}\mathcal R=L^2-4$ and encodes four static logical qubits with distance $d=L/2$.
\end{proposition}

\begin{figure}[t]
    \centering
    \includegraphics[width=0.95\columnwidth]{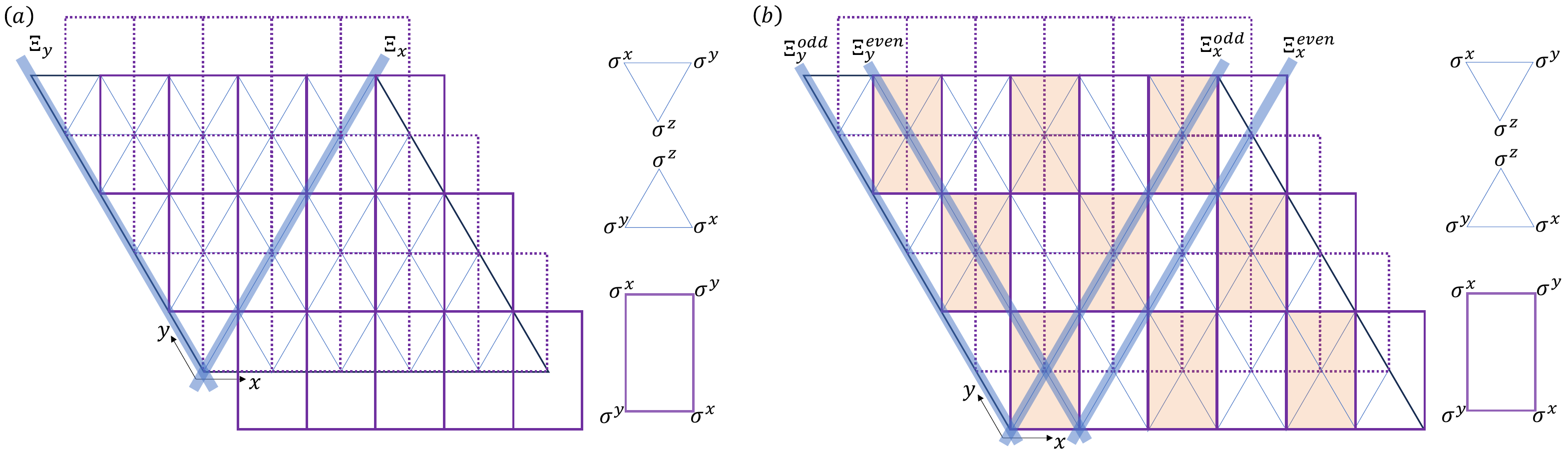}
    \caption{
    Parity-dependent structure of the one-step rectangle code.
    (a) For odd $L$, the periodic identification along the $y$ direction identifies the dashed rectangle pattern at the upper boundary with the solid pattern at the lower boundary, so all rectangle stabilizers belong to a single connected sublattice.
    Accordingly, multiplying all rectangle operators gives the unique global constraint.
    The blue noncontractible lines indicate a convenient logical Pauli pair, $\Xi_y$ and $\Xi_x$.
    (b) For even $L$, the dashed and solid rectangle patterns remain decoupled and form two independent rectangle sublattices.
    For the solid-line sublattice, the shaded rectangles multiply to the identity, while the remaining unshaded solid-line rectangles give a second independent identity.
    The dashed-line sublattice satisfies the analogous pair of constraints.
    The blue noncontractible lines indicate representative logical operators confined to one sublattice and one line class, for example ${\Xi}_y^{\rm odd}$, ${\Xi}_y^{\rm even}$, ${\Xi}_x^{\rm odd}$, and ${\Xi}_x^{\rm even}$.
    Thus each sublattice contributes two logical qubits, giving four logical qubits in total.
    }
    \label{figS:T1logical}
\end{figure}

\subsection{Odd $L$}
\label{sec:SMT1odd}

For odd $L$, the periodic identification along the $y$ direction connects the two apparent rectangle patterns into a single sublattice.
As illustrated in Fig.~\ref{figS:T1logical}(a), a dashed rectangle at the upper boundary is identified, after winding around the periodic direction, with a solid rectangle at the lower boundary.
Thus, all $L^2$ rectangle operators belong to one connected stabilizer sublattice.

We first determine the relations among the rectangle operators.
The product of all $L^2$ rectangles is the identity, giving one global constraint.
We now show that this is the only nontrivial relation.
Suppose that a product of a nonempty subset of rectangle operators is the identity, and choose one rectangle appearing in the product.
Cancellation at one of its corner qubits requires the neighboring rectangle along the same diagonal direction that shares this corner to appear as well.
Applying the same argument successively along the diagonal forces the entire strip of rectangles in that direction to be included.
The product of all rectangles in such a strip leaves the two noncontractible lines separated by two lattice spacings, $\Xi_\alpha^j\Xi_\alpha^{j+2}$.
For the full product to be the identity, the remaining line $\Xi_\alpha^{j+2}$ must in turn be cancelled by the neighboring strip, whose product contains $\Xi_\alpha^{j+2}\Xi_\alpha^{j+4}$.
Repeating this argument forces the strips labeled by $j,j+2,j+4,\ldots$ to appear.
Because $L$ is odd, repeated translation by two visits every line modulo $L$.
Hence, every strip, and therefore every rectangle, must be included.
Thus the product of all rectangles is the unique nontrivial relation, so $\operatorname{rank}\mathcal R=L^2-1$ and the code encodes one logical qubit.

A convenient logical Pauli pair is given by the two noncontractible loops shown in Fig.~\ref{figS:T1logical}(a),
\begin{equation}
    (\widetilde X,\widetilde Z)=(\Xi_x,\Xi_y).
\label{eq:SMT1oddLogical}
\end{equation}
Both commute with every rectangle stabilizer and anticommute with one another.

We next construct the stabilizers that translate these logical lines.
For either of the two line orientations appearing in Eq.~\eqref{eq:SMT1oddLogical}, multiplying the rectangle operators along the corresponding strip gives the product of two lines separated by two lattice spacings, $\Xi_\alpha^j\Xi_\alpha^{j+2}\in\mathcal R$.
Because $L$ is odd, translation by two visits every line modulo $L$.
More explicitly, since $2(L+1)/2=1\pmod L$, multiplying the successive step-two pairs gives
\begin{equation}
    \prod_{r=0}^{(L-1)/2}
    \left(\Xi_\alpha^{j+2r}\Xi_\alpha^{j+2r+2}\right)
    =
    \Xi_\alpha^j\Xi_\alpha^{j+1},
\label{eq:SMT1oddDoubleLine}
\end{equation}
where all intermediate line operators cancel.
Thus the nearest-neighbor double-line operator $\Xi_\alpha^j\Xi_\alpha^{j+1}$ also belongs to $\mathcal R$.

Multiplying a logical line by these nearest-neighbor double-line stabilizers translates it successively through all $L$ parallel positions.
Therefore each logical Pauli has $L$ mutually disjoint equivalent representatives.
Lemma~\ref{lem:SMdisjoint} gives $d\ge L$, while the noncontractible loops themselves have weight $L$.
Hence $d=L$.

\subsection{Even $L$}
\label{sec:SMT1even}

For even $L$, the periodic identification no longer connects the dashed and solid rectangle patterns.
As shown in Fig.~\ref{figS:T1logical}(b), they remain decoupled and form two independent rectangle sublattices, which we denote by $\mathcal R_{\rm s}$ and $\mathcal R_{\rm d}$.

We first determine the relations within the solid-line sublattice $\mathcal R_{\rm s}$.
Suppose that a nonempty product of solid-line rectangle operators is the identity and choose one rectangle appearing in the product.
As in the odd-$L$ case, cancellation at one of its corner qubits forces the neighboring rectangle along the same diagonal direction to appear.
Repeating this argument along the diagonal forces all rectangles of the same rectangle sublattice within that strip to be included.
The product of these rectangles gives a pair of noncontractible lines separated by two lattice spacings.
Cancellation of the second line then forces the corresponding strip of the same rectangle sublattice, displaced by two lattice spacings, to appear as well.
Thus the relation propagates through the sequence of strips $j,j+2,j+4,\ldots$ within the same rectangle sublattice.

For even $L$, translation by two does not visit all strips, but instead preserves their parity.
Thus the propagation fills either the even-strip class or the odd-strip class, but does not connect the two.
These two classes are represented by the shaded and unshaded solid-line rectangles in Fig.~\ref{figS:T1logical}(b).
Multiplying all shaded rectangles gives one identity, while multiplying all unshaded solid-line rectangles gives a second identity.
The propagation argument above also shows that any nontrivial relation is generated by these two, so there are exactly two independent constraints in $\mathcal R_{\rm s}$.
The dashed-line sublattice $\mathcal R_{\rm d}$ obeys the same two constraints by the analogous construction.

Since each rectangle sublattice contains $L^2/2$ generators, we obtain
\begin{equation}
    \operatorname{rank}\mathcal R_{\rm s}=\operatorname{rank}\mathcal R_{\rm d}=\frac{L^2}{2}-2, \qquad \operatorname{rank}\mathcal R=L^2-4.
\label{eq:SMT1evenRank}
\end{equation}
Hence the even-$L$ rectangle code encodes four logical qubits.

The blue lines in Fig.~\ref{figS:T1logical}(b) show a convenient logical basis for one of the two rectangle sublattices.
Because this sublattice separates into even and odd strip classes, the noncontractible loops split into the restricted operators $\Xi_x^{\rm even}$, $\Xi_x^{\rm odd}$, $\Xi_y^{\rm even}$, and $\Xi_y^{\rm odd}$.
The even-sector operators $(\Xi_x^{\rm even},\Xi_y^{\rm even})$ form one logical Pauli pair, while the odd-sector operators $(\Xi_x^{\rm odd},\Xi_y^{\rm odd})$ form the second.
Within each pair, the two operators anticommute because they have odd intersection parity, whereas operators belonging to different pairs commute.
Each restricted loop is supported on one of the two parity sectors of the corresponding noncontractible line and therefore contains $L/2$ qubits.
The other rectangle sublattice has the analogous two logical Pauli pairs.
Together, the two sublattices encode four logical qubits.

The same strip products used above also relate translated logical representatives.
For a fixed orientation and parity sector, they give $\Xi_\alpha^j\Xi_\alpha^{j+2}\in\mathcal R$.
Because $L$ is even, repeated translation by two visits exactly the $L/2$ lines of the same parity class.
Thus each logical generator has $L/2$ mutually disjoint equivalent representatives.
Every nontrivial logical Pauli operator anticommutes with at least one generator of the complete logical basis, so Lemma~\ref{lem:SMdisjoint} gives $d\ge L/2$.
Conversely, the restricted logical loops displayed in Fig.~\ref{figS:T1logical}(b) have weight $L/2$, and therefore
\begin{equation}
    d=\frac{L}{2}.
\label{eq:SMT1evenDistance}
\end{equation}

This completes the proof of Proposition~\ref{prop:SMT1}.

\section{Local Majorana realization with one open direction}

We finally give the direct microscopic Majorana realization of the static code and the suitable measurement protocols. 
Each unit cell of the honeycomb lattice contains two Majorana modes $\alpha_i$ and $\beta_i$. With the ordering shown in Fig.~\ref{fig:Majorana_rep}(a) of the main text, they are related to the effective spin operators by
\begin{equation}
\alpha_i=\sigma_i^x\prod_{k<i}(-\sigma_k^z),\qquad
\beta_i=\sigma_i^y\prod_{k<i}(-\sigma_k^z).
\label{eq:SM_JW}
\end{equation}
An individual Majorana operator therefore carries a Jordan--Wigner string whose support depends on the chosen ordering.

For the local triangle gauge generators, however, these strings cancel. With the vertex convention of Fig.~\ref{fig:Majorana_rep}(a), the two triangle generators become
\begin{equation}
A_\Delta=\beta_i\alpha_i\alpha_j\alpha_k,\qquad
B_\nabla=\alpha_k\beta_i\beta_j\beta_k,
\label{eq:SM_triangle_majorana}
\end{equation}
up to an irrelevant overall phase. Thus both triangle gauge generators have local four-Majorana parity representations.

On a fully periodic system, a linear Jordan--Wigner ordering necessarily closes across a seam.
Triangle generators crossing this seam generally acquire a nonlocal Jordan--Wigner string in the Majorana representation. 
To avoid this obstruction, we open one spatial direction and retain periodic boundary conditions only in the other direction, using the ordering shown in Fig.~\ref{fig:Majorana_rep}(a) in the main text.
The resulting cylindrical geometry contains $L(L-1)$ local $A_\Delta$ generators and $L(L-1)$ local $B_\nabla$ generators, each represented by a four-Majorana parity operator.

These local generators are sufficient to reproduce the same abstract gauge group as in the periodic spin description. 
Indeed, the independent generating set constructed in Proposition~\ref{prop:SMgauge_rank} can be chosen entirely from the retained local triangles. 
It contains $L(L-1)$ independent $A$-type generators and $(L-1)^2$ independent $B$-type generators, giving $\operatorname{rank}\mathcal G=(2L-1)(L-1)$.
The omitted boundary-crossing triangle operators are consequently generated by products of the retained local generators. 
Thus opening one spatial direction changes the microscopic realization but not the gauge-group algebra or the encoded subsystem-code structure.

We next consider the measurement operators. 
For each parallelogram that remains local after opening the boundary, substituting Eq.~\eqref{eq:SM_JW} cancels the Jordan--Wigner strings. 
The local $x$- and $y$-type parallelograms reduce to four-Majorana parity operators, whereas the local $z$-type parallelograms contain eight Majoranas,
\begin{equation}
\mathcal P_x,\mathcal P_y:\ 4\ {\rm Majoranas},\qquad
\mathcal P_z:\ 8\ {\rm Majoranas}.
\label{eq:SM_parallelogram_majorana}
\end{equation}
Likewise, an individual local rectangle operator reduces to a four-Majorana parity operator.
Importantly, however, locality of an individual measurement operator does not by itself imply that the complete periodic measurement group remains locally measurable after opening the boundary.

With the opening in Fig.~\ref{fig:Majorana_rep}(a) of the main text, the $L(L-1)$ $\mathcal P_x$ and $\mathcal P_z$ parallelograms in the first $L-1$ columns remain locally realizable, whereas one boundary strip of the $\mathcal P_y$ family is removed, leaving only $L(L-2)$ local $y$-type parallelograms. 
Let $\mathcal P_y^{\rm loc}$ denote the subgroup generated by these remaining operators. 
They are independent by a boundary-peeling argument.
Starting from an open boundary, cancellation of the Pauli operators on the exposed boundary sites forces all parallelograms in the outermost retained strip to be absent from any relation.
Repeating this argument successively through the bulk gives
\begin{equation}
\operatorname{rank}\mathcal P_y^{\rm loc}=L(L-2).
\label{eq:SM_Pyloc_rank}
\end{equation}

To show that the retained local measurements can be completed by the double-loop sector, choose the $(L-1)(L-2)$ local $y$-type parallelograms obtained by omitting one complete row from $\mathcal P_y^{\rm loc}$, together with the standard double-loop basis $D_y^{01},D_y^{12},\ldots,D_y^{L-2,L-1}$.
We now show that this combined set is independent.

Suppose that a product of these operators is the identity. 
At the boundary site $(0,L-1)$, none of the selected local parallelograms acts, while among the double loops only $D_y^{01}$ acts.
Hence $D_y^{01}$ cannot appear in the relation.
The open boundary then exposes the first column of selected local parallelograms, and cancellation on its boundary sites forces all $L-1$ parallelograms in this column to be absent. 
After removing this column, the same argument applies to the next boundary site, forcing $D_y^{12}$ to be absent, followed by all parallelograms in the next column. 
Repeating this procedure successively from left to right eliminates every double loop and every selected local parallelogram.
Thus the combined set is independent.

It contains $(L-1)(L-2)+(L-1)=(L-1)^2$
independent operators. 
Since all of them belong to $\mathcal P_y$ and $\operatorname{rank}\mathcal P_y=(L-1)^2$, they form a complete basis of $\mathcal P_y$. 
Consequently,
\begin{equation}
\mathcal P_y=\langle\mathcal P_y^{\rm loc},\mathcal D_y\rangle.
\label{eq:SM_Pyloc_completion}
\end{equation}
Equivalently, the retained local measurements already contain $L-2$ independent combinations of the $y$-type double loops. 
For example, the product of the $L$ local parallelograms in each retained complete strip gives a neighboring product $D_y^jD_y^{j+1}$ with $j=0,\cdots,L-3$.
Besides, any operator generated by the full $\mathcal P_y$ family can be reconstructed from the retained local measurements together with the $y$-type double loops. 
This will be important below for the even-$L$ three-step protocol, where the preceding rounds already supply $\mathcal D_y$. 
For example, the operator $F_y\in\mathcal P_y$ appearing in the logical update can be written as a product of elements of $\mathcal P_y^{\rm loc}$ and $\mathcal D_y$, so its eigenvalue is determined by the local $\mathcal P_y$ outcomes together with the inherited double-loop eigenvalues.

Since $\mathcal D_y\subset Z(\mathcal G)$, every element of $\mathcal G$ already commutes with the additional generators in $\mathcal D_y$. 
Consequently,
\begin{equation}
\mathcal G\cap C(\mathcal P_y^{\rm loc})=\mathcal G\cap C(\mathcal P_y).
\label{eq:SM_local_centralizer}
\end{equation}
This observation gives a local realization of the three-step protocol for even $L$. 
We cyclically start the period with the sequence
\begin{equation}
\mathcal P_z\longrightarrow\mathcal P_x\longrightarrow\mathcal P_y^{\rm loc}.
\label{eq:SM_Majorana_T3}
\end{equation}
Starting from a trivial ISG, the first two measurements give
\begin{equation}
\mathcal S_0=\mathcal P_z,\qquad
\mathcal S_1=\langle\mathcal P_x,\mathcal D_z,\mathcal D_y\rangle,
\label{eq:SM_Majorana_T3_first}
\end{equation}
where the second equality follows from Proposition~5 for even $L$. 
Since $\mathcal S_1\subset\mathcal G$, Eq.~\eqref{eq:SM_local_centralizer} gives
\begin{equation}
\mathcal S_2=\langle\mathcal P_y^{\rm loc},
\mathcal S_1\cap C(\mathcal P_y^{\rm loc})\rangle=\langle\mathcal P_y,\mathcal D_x,\mathcal D_z\rangle.
\label{eq:SM_Majorana_T3_second}
\end{equation}
Thus the truncated local $\mathcal P_y$ round produces exactly the same ISG as the full periodic $\mathcal P_y$ measurement.
Subsequent $\mathcal P_z$ and $\mathcal P_x$ rounds reproduce the remaining steady ISGs derived in Sec.~\ref{sec:SMT3}. 
Hence the even-$L$ $T=3$ protocol, including the partially dynamical case $L=4n$, admits a direct local Majorana realization.

For odd $L$, the same mechanism does not apply. 
After the preceding $\mathcal P_z$ and $\mathcal P_x$ rounds,
\begin{equation}
\mathcal S_1=\langle\mathcal P_x,\mathcal D_z\rangle,
\end{equation}
so the missing independent part of the $\mathcal D_y$ sector is not inherited. 
Measuring $\mathcal P_y^{\rm loc}$ therefore gives
\begin{equation}
\mathcal S_2=\langle\mathcal P_y^{\rm loc},\mathcal D_x,\mathcal D_z\rangle,\qquad\operatorname{rank}\mathcal S_2=L^2-2,
\label{eq:SM_Majorana_odd}
\end{equation}
rather than the rank-$(L^2-1)$ periodic ISG.
Additional boundary measurements are therefore required to reproduce the odd-$L$ three-step code.

The even-$L$ two-step protocol can be realized more directly by choosing the cyclically equivalent measurement sequence
\begin{equation}
\mathcal P_z\longrightarrow\mathcal P_x\longrightarrow\mathcal P_z\longrightarrow\cdots .
\label{eq:SM_majorana_T2}
\end{equation}
Both measurement families are locally realizable on the cylinder.
Their steady ISGs alternate as
\begin{equation}
\langle\mathcal P_x,\mathcal D_y,\mathcal D_z\rangle
\longleftrightarrow
\langle\mathcal P_z,\mathcal D_x,\mathcal D_y\rangle.
\end{equation}
By cyclic symmetry, this is equivalent to the two-step protocol in the main text and therefore has the same two static logical qubits and distance $d=L$. 
The local $\mathcal P_x$ measurements involve four Majoranas and the local $\mathcal P_z$ measurements involve eight.

Finally, although an individual rectangle operator has a local four-Majorana representation, simply truncating the rectangle family at the open boundary does not reproduce the periodic $T=1$ stabilizer code.
Opening one direction leaves only $L(L-1)$ local rectangles, so their rank is at most $L^2-L$. 
By contrast, the periodic rectangle group has
\begin{equation}
\operatorname{rank}\mathcal R=
\begin{cases}
L^2-1,&L\ {\rm odd},\\
L^2-4,&L\ {\rm even}.
\end{cases}
\end{equation}
Thus the bulk local rectangles alone do not supply the stabilizer rank required for the periodic code in general. 
A local realization of the remaining boundary cases therefore requires additional boundary checks, whose construction we leave for future work.

Consequently, with the present open-boundary construction, direct local microscopic Majorana realizations are established for the even-$L$ two- and three-step protocols. 
Both use only local four- and eight-Majorana parity measurements, while extending the construction more generally requires an appropriate boundary completion.

\section{Summary of the measurement protocols}
\label{sec:SMsummary}

For convenience, we summarize the logical structure of all measurement protocols in Table~\ref{tab:SMsummary}.
Here $(k_s,k_p)$ denotes the numbers of static and partially dynamical logical qubits, respectively.
The operators shown are convenient representatives of the corresponding logical Pauli classes, and equivalent representatives may differ by elements of the instantaneous stabilizer group.

\begin{table*}[t]
\centering
\begin{tabular}{c c c c c}
\toprule
Protocol
& System size
& $(k_s,k_p)$
& Logical Pauli representatives
& $d$
\\
\midrule

$T=3$
& $L$ odd
& $(1,0)$
& $(\Xi_x,\Xi_z)$
& $L$
\\[2pt]

$T=3$
& $L=4n+2$
& $(2,0)$
& $(\Xi_x,\Xi_z)$;\quad$(\Omega_x,H_L)$
& $L$
\\[2pt]

$T=3$
& $L=4n$
& $(1,1)$
& $(\Xi_x,\Xi_z),\quad (\widetilde X_2^{(t=0,1,2)},H_L)$
& $L$
\\[4pt]

$T=2$
& $L=4n+2$
& $(2,0)$
& $(\Xi_x,\Xi_z)$;\quad $(\Omega_x,H_L)$
& $L$
\\[2pt]

$T=2$
& $L=4n$
& $(2,0)$
& $(\Xi_x,\Xi_z)$;\quad$(\Xi_{h,x},H_L)$
& $L$
\\[2pt]

$T=1$
& $L$ odd
& $(1,0)$
& $(\Xi_x,\Xi_y)$
& $L$
\\[2pt]

$T=1$
& $L$ even
& $(4,0)$
& $\begin{array}{l}
(\Xi_{x,{\rm s}}^{\rm even},
 \Xi_{y,{\rm s}}^{\rm even}),\
(\Xi_{x,{\rm s}}^{\rm odd},
 \Xi_{y,{\rm s}}^{\rm odd}),\
(\Xi_{x,{\rm d}}^{\rm even},
 \Xi_{y,{\rm d}}^{\rm even}),\
(\Xi_{x,{\rm d}}^{\rm odd},
 \Xi_{y,{\rm d}}^{\rm odd})
\end{array}$
& $L/2$
\\
\bottomrule
\end{tabular}
\caption{
Summary of the logical structure of the one-, two-, and three-step measurement protocols.
Here $(k_s,k_p)$ denotes the numbers of static and partially dynamical logical qubits, respectively.
For the three-step protocol with $L=4n$, $\widetilde Z_2=H_L$ remains static, while its conjugate follows the three-round orbit in Eq.~\eqref{eq:SMDynamicOrbit}.
For the even-$L$ one-step code, the subscripts ${\rm s}$ and ${\rm d}$ denote the solid- and dashed-line rectangle sublattices, respectively, and the superscripts ${\rm even}$ and ${\rm odd}$ denote the two line-parity sectors within each sublattice.
}
\label{tab:SMsummary}
\end{table*}

\end{widetext}

\end{document}